\documentclass{article}

\usepackage[english]{babel}
\usepackage{setspace}
\usepackage{tikz}
\usepackage{arydshln}
\usepackage{relsize}

\usepackage[letterpaper,top=2cm,bottom=2cm,left=3cm,right=3cm,marginparwidth=1.75cm]{geometry}

\usepackage{amsmath,extarrows}
\usepackage{amssymb}
\usepackage{amsthm}
\usepackage{amsfonts}
\usepackage{graphicx}
\usepackage{cmll}
\usepackage{relsize}
\usepackage[colorlinks=true, allcolors=blue]{hyperref}
\usepackage{natbib}
\usetikzlibrary {arrows.meta}
\usepackage{bussproofs}
\title{Game Semantics for Higher-Order Unitary Quantum Computation}
\author{Samson Abramsky\thanks{
S.Abramsky@cs.ucl.ac.uk \\
University College London, Computer Science, Gower Street, London WC1E 6BT
United Kingdom.} and Radha Jagadeesan
\thanks{rjagadee@depaul.edu \\
School of Computing, 243 S Wabash, Chicago, IL 60540, USA.
}}

\begin{document}
\maketitle
\newcommand{\icoh}{\smile}

\newcommand{\id}{{\tt id}}
\newcommand{\twist}{\tt twist}
\newcommand{\lltwist}{\tt twist_{\llwith}}
\newcommand{\ptwist}{\tt twist_{\lplus}}
\newcommand{\tensor}{\bigotimes}
\newcommand{\ltensor}{\otimes}
\newcommand{\linimpl}{\multimap}
\newcommand{\parc}{\bigparr}
\newcommand{\bp}[1]{\left( #1 \right)}
\newcommand{\adj}[1]{#1^{\dag}}

\newcommand{\llwith}{\&}

\newcommand{\Games}{{\cal G}}
\newcommand{\agame}{G}
\newcommand{\bgame}{H}
\newcommand{\cgame}{K}

\newcommand{\moves}{M}

\newcommand{\lpo}{\sqsubseteq}
\newcommand{\posn}{P}
\newcommand{\maximal}[1]{#1^m}
\newcommand{\posnmax}{\posn{}^m}
\newcommand{\win}{W}
\newcommand{\PO}{{}^{{\small PO}}_{{\small QA}}\lambda}

\newcommand{\dom}[1]{ D_{#1}}
\newcommand{\od}{{\tt D^O}}
\newcommand{\pd}{{\tt D^P}}
\newcommand{\ad}{D_{\agame}}
\newcommand{\bd}{D_{\bgame}}
\newcommand{\cd}{D_{\cgame}}
\newcommand{\aod}{D{}^O_{\agame}}
\newcommand{\aodp}{D{}^O_{\agame'}}
\newcommand{\apd}{D{}^P_{\agame}}
\newcommand{\apdp}{D{}^P_{\agame'}}
\newcommand{\bod}{D{}^O_{\bgame}}
\newcommand{\bpd}{D{}^P_{\bgame}}
\newcommand{\cpd}{D{}^P_{\cgame}}
\newcommand{\cod}{D{}^O_{\cgame}}
\newcommand{\ats}{v}
\newcommand{\bts}{w}

\newcommand{\amoves}{\moves_{\agame}}
\newcommand{\bmoves}{\moves_{\bgame}}
\newcommand{\cmoves}{\moves_{\cgame}}
\newcommand{\aposn}{\posn_{\agame}}
\newcommand{\bposn}{\posn_{\bgame}}
\newcommand{\aposnmax}{\posnmax_{\agame}}
\newcommand{\aposnmaxr}{\posnmax_{\agamer}}
\newcommand{\bposnmax}{\posnmax_{\bgame}}
\newcommand{\aPO}{{}^{{\small PO}}_{QA}\lambda{}_{\agame}}
\newcommand{\bPO}{{}^{{\small PO}}_{QA}\lambda{}_{\bgame}}

\newcommand{\apos}{s}
\newcommand{\bpos}{t}
\newcommand{\cpos}{u}

\newcommand{\inv}[1]{#1^{-1}}

\newcommand{\bool}{{\tt Bool}}
\newcommand{\qbit}{{\tt QBit}}
\newcommand{\one}{{\bf 1}}

\newcommand{\oans}{{\tt oans}}
\newcommand{\pans}{{\tt pans}}
\newcommand{\answers}{{\tt ans}}

\newcommand{\fn}[1]{f_{#1}}

\newcommand{\amove}{m}
\newcommand{\bmove}{n}
\newcommand{\hi}[1]{{\tt D}_{#1}}
\newcommand{\ahi}{{\tt D}_{\agame}}
\newcommand{\bhi}{{\tt D}_{\bgame}}

\newcommand{\m}[2]{#1 \cdot #2}
\newcommand{\mstar}[2]{#1 \cdot #2}
\newcommand{\minfty}[2]{#1 \cdot #2}

\newcommand{\complex}{{\cal C}}
\newcommand{\abs}[1]{\lvert #1 \rvert}
\newcommand{\ac}{z}
\newcommand{\bc}{u}

\renewcommand{\fCenter}{\mbox{$\vdash$}}
\newcommand{\aForm}{A}
\newcommand{\bForm}{B}
\newcommand{\cForm}{C}
\newcommand{\amForm}{\Gamma}
\newcommand{\bmForm}{\Delta}

\newcommand{\str}[1]{\sigma_{#1}}
\newcommand{\astr}{\sigma}
\newcommand{\bstr}{\tau}
\newcommand{\conj}[1]{#1^{\star}}
 
\newtheorem{thm}{Theorem} 
\newtheorem{defn}[thm]{Definition}
\newtheorem{lemma}[thm]{Lemma} 
\newtheorem{example}[thm]{Example} 
\newtheorem{corollary}[thm]{Corollary}
\newtheorem{comment}[]{Tangential Comment}

\newcommand{\V}[1]{{\mathcal V}(#1)}
\newcommand{\Hi}[1]{{\mathcal H}(#1)}
\newcommand{\carr}[1]{\left| #1 \right| }
\newcommand{\norm}[1]{\left|\left| #1 \right| \right| }

\newcommand{\rev}[1]{{\mathcal R}(#1)}
\newcommand{\pair}[2]{\langle #1, #2 \rangle}
\newcommand{\power}[1]{{\cal P}(#1)}
\newcommand{\kernel}[1]{{\tt ker}(#1)}
\newcommand{\rng}[1]{{\tt rng}(#1)}
\newcommand{\EX}{{\tt EX}}
\newcommand{\pcomp}{\left| \right|}
\newcommand{\rest}{\downharpoonright}
\newcommand{\red}[1]{\overline{#1}}

\newcommand{\true}{{\tt tt}}
\newcommand{\false}{{\tt ff}}
\newcommand{\cnot}{{\tt CNOT}}

\newcommand{\atv}{v}
\newcommand{\btv}{w}
\newcommand{\pref}[1]{{\tt Pref}(#1)}
\newcommand{\suff}[1]{{\tt Suff}(#1)}

\newcommand{\acr}{\coh_{\agame}}
\newcommand{\bcr}{\coh_{\bgame}}

\newcommand{\bgg}{{\mathcal G}}
\newcommand{\iso}{\ \raisebox{0.5ex}{\ensuremath{\stackrel{\sim}{\rule{16pt}{0.7pt}}}\ \,}}

\newcommand{\vg}{{\cal V}}
\newcommand{\qg}{{\cal U}}
\newcommand{\qcg}{{\cal CP}}
\newcommand{\R}{{\cal R}}

\newcommand{\ndim}{{\bf n}}
\newcommand{\five}{{\bf 5}}

\newcommand{\adjadj}[1]{#1^{\dagger\dagger}}
\newcommand{\mat}[1]{M_{#1}}
\newcommand{\amat}{\mat{\astr}}
\newcommand{\bmat}{\mat{\bstr}}

\newcommand{\lplus}{\bigoplus}
\newcommand{\assoc}{{\tt assoc}}
\newcommand{\llassoc}{{\tt assoc}_{\llwith}}
\newcommand{\passoc}{{\tt assoc}_{\lplus}}
\newcommand{\app}{{\tt app}}
\newcommand{\injl}{{\tt injl}}
\newcommand{\injr}{{\tt injr}}
\newcommand{\lsumstr}[2]{[#1,#2]}
\newcommand{\distts}{{}^{\tiny \tensor\lplus}{\tt dist}}
\newcommand{\distiw}{{}^{\linimpl\llwith}{\tt dist}}

\newcommand{\trace}[2]{{\tt Tr}{}^{#1}_{#2}}
\newcommand{\cntrl}[1]{{\tt cntrl}(#1)}

\newcommand{\avec}{\Delta} 
\newcommand{\bvec}{\Gamma}

% \tableofcontents

\begin{abstract}
We develop a symmetric monoidal closed category of games, incorporating sums and products, to model quantum computation at higher types. This model is expressive, capable of representing all unitary operators at base types. It is compatible with base types and realizable by unitary operators.
\end{abstract}

\section{Introduction}
We understand how to model quantum control in Turing machines~\citep{doi:10.1137/S0097539796300921} and how to model higher-order classical computation over quantum data~\citep{Selinger2005,Sel14}. However, we seem to lack similar agreement and understanding on quantum, structured, higher-order control in the style of functional programming. This is the problem posed by \citeauthor{Sel2004}; what is a semantic model of higher-order quantum computation? In this paper, we provide a candidate solution, albeit without accommodating measurement.  Our starting points is $\bgg$, a new (deterministic) games semantics for Intuitionist, Multiplicative($\tensor, \linimpl)$, Additive ($\llwith,\lplus$) Linear Logic(IMALL).  We build on $\bgg$ to describe the following categories of games.
\begin{itemize}
\item $\vg$ is universal for all linear maps on $n$-ary booleans.  $\vg$ models IMLL and $\lplus$.

\item $\qg$, a subcategory of $\vg$, that permits only unitary morphisms.  $\qg$ is universal for unitaries  on $n$-ary tensor of qbits, and is a model of IMLL, with an additional monoidal operation $\lplus$.
\end{itemize}
Both categories support appropriate distributivity between the monoidal operations.

Consider the desiderata outlined by \citet{Sel2004} for a solution, which we paraphrase as follows:

\begin{description}
\item[Universality:] The model must support all unitary operations over the $n$-ary tensor product of Qbits.
\item[Compatibility:] The model should align with base types, ensuring that  all type constructors, such as tensor ($\otimes$), apply uniformly across types.
\item[Realizability:] The morphisms should be implementable as quantum processes.
\end{description}

From existing formalisms, the ZX-calculus~\citep{ZX} achieves universality for all linear maps on booleans, thereby (intentionally) exceeding the scope of quantum processes. Selinger's analysis shows that Quantum Coherence spaces~\citep{Girard2004-GIRBLA} lacks both universality and realizability, whereas Selinger's normed cones~\citep{Sel2004} fail to meet compatibility with base types.   ~\citet{10.1145/3632861} meets all listed criteria, albeit only within a first-order framework.

Our understanding of deterministic higher-order computation emerges from the tight linkage between types/formulas and programs/proofs. Acknowledging that quantum processes will inevitably yield new proofs, we adhere to the following principle to preserve the connection with logic:

\begin{description}
\item[Conservativity:] New theorems should not emerge among old formulas.
\end{description}

While we do not exclude the possibility of introducing novel type/formula constructors unique to the quantum domain (e.g., see~\citet{FKS2022-lindep-long}), the validity of this principle is debatable, as it seemingly prioritizes existing logic. Nonetheless, we cautiously adopt this principle, viewing it as a constructive limitation providing guardrails on quantum higher-order control flow.

We informally discuss the key conceptual elements of our design within this context. As motivation, we recall that higher-order combinators are subtle, even in the more restricted framework of reversible computation. Consider application, $\app$, traditionally depicted through the directional flow of data among function, argument, and result, as illustrated by the arrows in the image below.

\vspace{5pt}

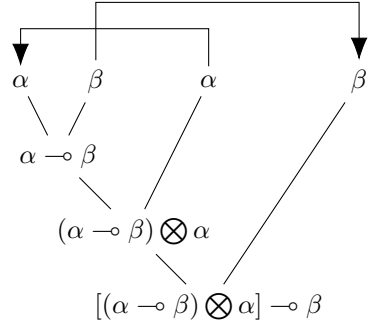
\begin{minipage}{\textwidth}
\centering
\begin{tikzpicture}[node distance=1cm and 1cm]
    \node (a1) at (1,-1.0) {$[(\alpha \linimpl\beta) \tensor \alpha]\linimpl \beta$};  
    \node (a) at (0,0) {$(\alpha\linimpl\beta) \tensor \alpha$};
    \node (b) at (-1,1.0) {$\alpha\linimpl\beta$};
    % \node (c) at (1,1.5) {$\alpha\tensor\alpha$};
    \node (d) at (-1.5,2.0) {$\alpha$};
    \node (e) at (-0.5,2.0) {$\beta$};
    \node (f) at (1,2.0) {$\alpha$};
    \node (g) at (3,2.0) {$\beta$};

    % Lines
    \draw (a1) -- (a);
    \draw (a) -- (b);
    \draw (a) -- (f);
    \draw (b) -- (d);
    \draw (b) -- (e);
    \draw (a1) -- (g); 

% Line segments with three parts
     \draw[-{Latex[length=3mm]}] (f) -- ++(0, 0.7) -| (d);
     \draw[-{Latex[length=3mm]}] (e) -- ++(0, 1.05) -| (g);
\end{tikzpicture}
\end{minipage}

\vspace{5pt}

How does one consider $\app$ within a reversible framework? We must move beyond the first-order framework of (rig)-groupoids~\citep{10.1145/3632861}. Instead, we embrace the notion that reversible computing does not necessarily mean computations are invertible~\citep{HEUNEN2015217}, avoiding the pursuit of a witness in $\beta \linimpl (\alpha \linimpl \beta) \tensor \alpha$. Our goal is to also incorporate the flow of information in the direction opposite to that depicted in the figure above.

This perspective naturally leads to viewing $\app$ as a proof in multiplicative linear logic, as discussed in the general treatment of reversibility by~\citet{ABRAMSKY2005441}, which includes exponentials. Conceptually, this involves disregarding the directionality of the arrows in the figure above, a change that might seem trivial but is significant enough to warrant illustration in the subsequent diagram.

\vspace{10pt}

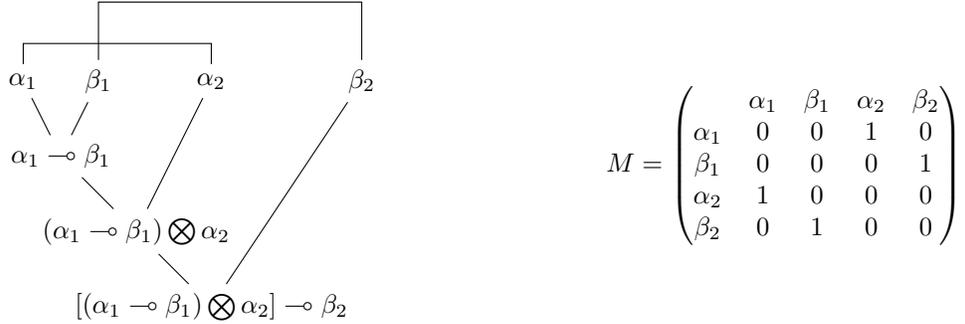
\begin{minipage}{0.6\textwidth}
\centering
\begin{tikzpicture}[node distance=1cm and 1cm]
    \node (a1) at (1,-1.0) {$[(\alpha_1\linimpl\beta_1) \tensor \alpha_2]\linimpl \beta_2$}; 
    \node (a) at (0,0) {$(\alpha_1\linimpl\beta_1) \tensor \alpha_2$};
    \node (b) at (-1,1.0) {$\alpha_1\linimpl\beta_1$};
    \node (d) at (-1.5,2.0) {$\alpha_1$};
    \node (e) at (-0.5,2.0) {$\beta_1$};
    \node (f) at (1,2.0) {$\alpha_2$};
    \node (g) at (3,2.0) {$\beta_2$};

    % Lines
    \draw (a1) -- (a);
    \draw (a) -- (b);
    \draw (a) -- (f);
    \draw (b) -- (d);
    \draw (b) -- (e);
    \draw (a1) -- (g); 

% Line segments with three parts
     \draw (f) -- ++(0, 0.5) -| (d);
     \draw (e) -- ++(0, 1.05) -| (g);
\end{tikzpicture}
\end{minipage}
\begin{minipage}{0.4\textwidth}
\centering
$
 M= \begin{pmatrix}
& \alpha_1 &  \beta_1& \alpha_2 &  \beta_2 \\
\alpha_1& 0 & 0 & 1& 0 \\
\beta_1 & 0 & 0 & 0 & 1 \\
\alpha_2&  1 & 0 & 0 & 0\\ 
\beta_2 &  0 & 1 & 0 & 0
\end{pmatrix}$
\end{minipage}

\vspace{5pt}
We distinguish occurrences with subscripts.  
In multiplicative linear logic, proofs are structured around axiom links that pair identical atoms of opposite polarity. The remainder of the proof is wholly defined by the formula and essentially acts as a verification of the proof's validity. It's important to note that the switchings in the correctness criteria are dictated by the types alone, as highlighted by~\citep{Danos1989-DANTSO-5}.  For $\app$, the axiom links are represented by the image and matrix $M$ above.  Since the axiom links are a bipartite matching, $M$ is a symmetric orthogonal matrix.

Building on the Geometry of Interaction~\citet{GIRARD1989221}, game semantics~\citet{BLASS92,AJ94} models axiom links as copycat strategies. These strategies entail copying moves from one side of the axiom link to the other, highlighting a bidirectional and symmetric exchange that underscores the reversible nature of the $\app$ combinator. Accordingly, $\app$--—along with all other proofs—--aligns with the core principle of game semantics:
\begin{center}
\emph{The semantics of a program is a set of traces}
\end{center}
It's noteworthy that this framework does not explicitly address higher-order structures, such as abstraction and higher-order functionals, a characteristic especially prevalent in the style of games proposed by~\citet{AJ94,AJMPCF}\footnote{Conversely, the games by~\citet{HYLAND2000285} and~\citet{DBLP:conf/lfcs/Nickau94} incorporate a “pointer” structure, which deviates from the linearity of traces.}. 

Our objective is to broaden this perspective to encompass more complex forms of quantum computation. Within a proof for a given formula, the modifiable element lies in the configuration of axiom links. To this end, we propose the following generalization:
\begin{center}
   \emph{perfect matching on bipartite graphs $\rightsquigarrow$ unitary flow}
\end{center}
To see this, consider $(\alpha\tensor\alpha) \linimpl \alpha \tensor \alpha$ and its two illustrative proofs, $\id$ and symmetry ($\twist$), operations, described below.  The matrices are written in block-structured form, with 
$I = \begin{pmatrix}  
 1 & 0\\ 
 0 & 1 
\end{pmatrix}$
and
$X = \begin{pmatrix} 
 0 & 1\\ 
 1 & 0 
\end{pmatrix}$,
and rows and columns are ordered $\alpha_1, \ldots, \alpha_4$.  

\vspace{10pt}

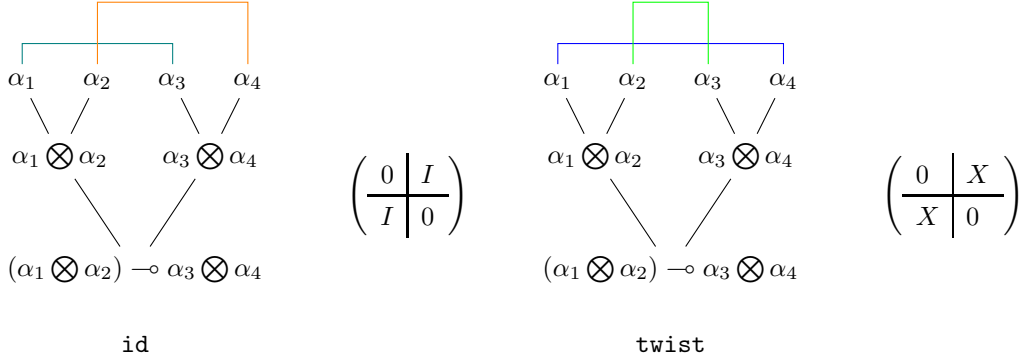
\begin{minipage}{.25\textwidth}
\centering
\begin{tikzpicture}[node distance=1cm and 1cm]
    \node (t) at (0,-1) {$\id$};
    \node (a) at (0,0) {$(\alpha_1\tensor\alpha_2) \linimpl \alpha_3 \tensor \alpha_4$};
    \node (b) at (-1,1.5) {$\alpha_1\tensor\alpha_2$};
    \node (c) at (1,1.5) {$\alpha_3\tensor\alpha_4$};
    \node (d) at (-1.5,2.5) {$\alpha_1$};
    \node (e) at (-0.5,2.5) {$\alpha_2$};
    \node (f) at (0.5,2.5) {$\alpha_3$};
    \node (g) at (1.5,2.5) {$\alpha_4$};

    % Lines
    \draw (a) -- (b);
    \draw (a) -- (c);
    \draw (b) -- (d);
    \draw (b) -- (e);
    \draw (c) -- (f);
    \draw (c) -- (g);

% Line segments with three parts
    % d to f
    \draw[teal] (d) -- ++(0, 0.5) -| (f);
    \draw[orange] (e) -- ++(0, 1.05) -| (g);
\end{tikzpicture}
\end{minipage}
\begin{minipage}{.2 \textwidth}
\renewcommand{\arraystretch}{1.3}
\centering
\[
\left(
 \begin{array}{c|c}
0 & I \\ \hline 
 I & 0 
\end{array}
\right) 
\]
\end{minipage}
\begin{minipage}{.25\textwidth}
\begin{tikzpicture}[node distance=1cm and 1cm]
     \node (t) at (0,-1) {$\twist$};
    \node (a) at (0,0) {$(\alpha_1\tensor\alpha_2) \linimpl \alpha_3 \tensor \alpha_4$};
    \node (b) at (-1,1.5) {$\alpha_1\tensor\alpha_2$};
    \node (c) at (1,1.5) {$\alpha_3\tensor\alpha_4$};
    \node (d) at (-1.5,2.5) {$\alpha_1$};
    \node (e) at (-0.5,2.5) {$\alpha_2$};
    \node (f) at (0.5,2.5) {$\alpha_3$};
    \node (g) at (1.5,2.5) {$\alpha_4$};

    % Lines
    \draw (a) -- (b);
    \draw (a) -- (c);
    \draw (b) -- (d);
    \draw (b) -- (e);
    \draw (c) -- (f);
    \draw (c) -- (g);

% Line segments with three parts
     \draw[blue] (d) -- ++(0, 0.5) -| (g) ;
     \draw[green] (e) -- ++(0, 1.05) -| (f);
\end{tikzpicture}
\end{minipage}
\begin{minipage}{.2\textwidth}
\renewcommand{\arraystretch}{1.3}
    \centering
    \[
    \left(
     \begin{array}{l|l}
  0 &  X\\ \hline
  X & 0
\end{array}
\right)
\]
\end{minipage}

Let us consider a new possible proof $\sqrt{\twist}$ whose axiom links are represented by the unitary, hence realizable, matrix (with same row, column conventions as above):

\vspace{10pt}
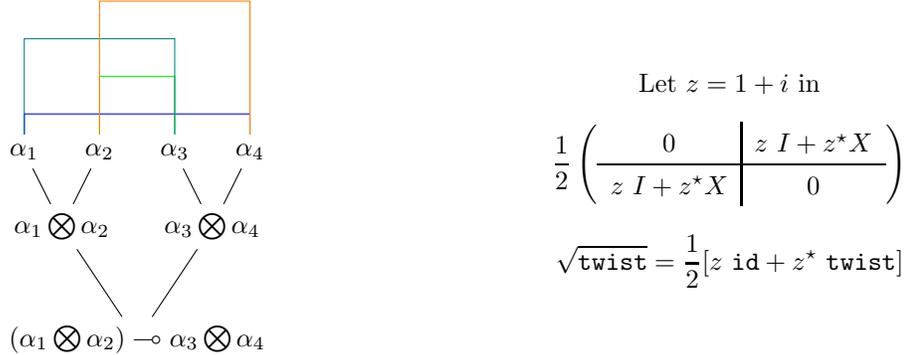
\begin{minipage}{.6\textwidth}
\centering
\begin{tikzpicture}[node distance=1cm and 1cm]
    \node (a) at (0,0) {$(\alpha_1\tensor\alpha_2) \linimpl \alpha_3 \tensor \alpha_4$};
    \node (b) at (-1,1.5) {$\alpha_1\tensor\alpha_2$};
    \node (c) at (1,1.5) {$\alpha_3\tensor\alpha_4$};
    \node (d) at (-1.5,2.5) {$\alpha_1$};
    \node (e) at (-0.5,2.5) {$\alpha_2$};
    \node (f) at (0.5,2.5) {$\alpha_3$};
    \node (g) at (1.5,2.5) {$\alpha_4$};

    % Lines
    \draw (a) -- (b);
    \draw (a) -- (c);
    \draw (b) -- (d);
    \draw (b) -- (e);
    \draw (c) -- (f);
    \draw (c) -- (g);

% Line segments with three parts
     \draw[blue] (d) -- ++(0, 0.5) -| (g) ;
     \draw[green] (e) -- ++(0, 1.0) -| (f);
     
     \draw[teal] (d) -- ++(0, 1.5) -| (f);
    \draw[orange] (e) -- ++(0, 2.) -| (g);
  
\end{tikzpicture}
\end{minipage}
\begin{minipage}{.4\textwidth}
\renewcommand{\arraystretch}{1.3}
    \centering
    Let $\ac = 1 + i$ in  
    \[\frac{1}{2}
    \left(
     \begin{array}{c|c}
  0 &  \ac\ I + \conj{\ac} X \\ \hline
 \ac\ I + \conj{\ac} X   & 0
\end{array}
\right)
\]
\[
\sqrt{\twist} = \frac{1}{2} [ \ac\ \id + \conj{\ac} \  \twist ]
\]
\end{minipage}

The matrix of $\sqrt{twist}$ is unitary, thus meeting \emph{Realizability}.  Assuming bilinearity of cut elimination:
\[\sqrt{\twist} ; \sqrt{\twist} = \twist \]
So, $\sqrt{\twist}$ is novel from the perspective of monoidal categories.  Furthermore, each (deterministic) summand\footnote{Contrast with probabilistic/quantum coherence space models~\cite{Girard2004-GIRBLA}, where only convex combinations of component elements arise in $\lplus$ types, and $\tensor$ is the closure of sums of tensors of component elements. } of $\sqrt{\twist}$ is verifiable by extant correctness criterion; thus meeting the \emph{conservativity} requirement.    

The presence of such morphisms hints at the possibilities available in quantum control.  For example, the two terms $t_1 = \lambda (f,g) \lambda x g(f(x))$ and $t_2 = \lambda (f,g) \lambda x f(g(x))$ are the two proofs of 
$[(\alpha \linimpl \alpha) \tensor \alpha \linimpl \alpha)] \linimpl (\alpha \linimpl \alpha)
$
that are connected as: 
\[
t_1 = (\alpha \linimpl \alpha) \tensor (\alpha \linimpl \alpha) \xlongrightarrow{\twist} (\alpha \linimpl \alpha) \tensor (\alpha \linimpl \alpha)  \xlongrightarrow{t_2} \alpha \linimpl \alpha
\]
In the quantum world, we can also construct:
\[
(\alpha \linimpl \alpha) \tensor (\alpha \linimpl \alpha) \xlongrightarrow{\sqrt{\twist}} (\alpha \linimpl \alpha) \tensor (\alpha \linimpl \alpha)  \xlongrightarrow{t_2} \alpha \linimpl \alpha
\]
where we use $\sqrt{\twist}$ at type $(\alpha \linimpl \alpha) \tensor (\alpha \linimpl \alpha)$.

In the paper, we show how to implement the idea semantically; roughly, performing the following substitutions;
\begin{alignat*}{4}
  \mbox{Formulas } & \rightsquigarrow\mbox{Games} \\
  \mbox{Proofs}    & \rightsquigarrow\mbox{(Reversible) Strategies} \\
  \mbox{Cut elimination} & \rightsquigarrow \mbox{Parallel composition + Hiding} \\
  \mbox{Generalized axiom links}  & \rightsquigarrow  \mbox{''Sums'' of strategies} 
\end{alignat*}
The first three elements above are standard to game semantics,~\citet{Abramsky_1997}; our novel contribution in this context is showing how to interpret additives reversibly, overcoming the difficulties described in~\citet{ABRAMSKY2005441}, and extending to the quantum realm.

The remainder of this paper is dedicated to formally elaborating on these concepts. We begin by introducing a new model of deterministic games ($\bgg$), tailored to Intuitionist Multiplicative Additive Linear Logic(IMALL). $\bgg$ is characterized by its history-sensitive strategies that observe distributivity—-- specifically, tensor over plus and implication over product. 

Subsequent sections of the paper expand this model to include appropriate sums, thereby enriching the category further. We introduce a category, $\vg$, that serves as a model for IMLL+$\lplus$, that is  universal in accommodating all linear operators on $n$-ary booleans. 

In the final sections, we explore reversible and unitary computation. We illustrate how a category of unitaries ($\qg$) can be defined as a subcategory of $\vg$. Notably, we identify a “first-order” fragment of $\qg$ that constitutes a rig-groupoid, complemented with the additional structure as outlined by~\cite{10.1145/3632861}. 

\section{Notation}
\begin{itemize}
    \item We will work with reflexive, binary coherence relations.  We use $\coh$ (resp. $\icoh$) for the coherence (resp. strict incoherence) relation.  Cliques are sets of elements that are pairwise coherent.  

   \item For any set $X$, $X^\star$ is the set of finite sequences of $X$.  

    We use juxtaposition for concatenation.  We will sometimes overload juxtaposition for concatenation when one or both of the arguments is a set to indicate the set derived by pointwise extensions. 
   
   We will consider prefix ordering $\preceq$, e.g., $\apos \preceq \apos \amove$.  For $S \subseteq X^{\star}$, we write $\pref{S}$ for the prefix closure of $S$.  
    
   For $x \in X, \apos \in X^{\star}$, we write $x \in \apos$ if $x$ is in $\apos$. For $S \subseteq X, \apos \in X^{\star}$, we write $S \subseteq \apos$ if every $x \in S$ is in $\apos$.
   
    \item Let $\complex$ be the complex numbers.  We use $\ac,\bc$ for complex numbers.
    
    \item For a finite set $X$, 
    $\Hi{X}$ is the finite dimensional Hilbert space with orthonormal basis $X$ over the field of complex numbers.  We use $\lplus$ (resp. $\tensor$) for the direct-sum (resp. tensor) of Hilbert spaces.   

    For a linear map $f$ between Hilbert spaces, $\adj{f}$ will stand for its adjoint.   We will use the matrix representation of linear maps to illustrate them; in that context, we use $\tensor$ to also stand for the matrix-tensor product. 

    In order to be able to relate to prior set based definitions on games, we use some notation to relate elements of $\Hi{X}$ to the subsets of $\complex \times X$.   Let $v \in \Hi{X} $.  We associate with $v$, the set of elements $\{ (\ac, x) \mid  \ac \not=0,  \ac = \mbox{inner product of } v, x \} $, and abuse notation sometimes, eg. $(\ac,x) \in v$.

    \item We use ${\cal C}(A,B)$ to refer to the morphisms between $A$ and $B$ in category ${\cal C}$.  This is especially useful for us because we sometimes view the same underlying morphisms in the context of different categories.  

    \item We use the language of traced monoidal categories~\citep{Joyal_Street_Verity_1996,SamsonTraces} to organize our presentation of composition. Two particular examples are particularly relevant to our presentation.  
    
    The category of relations with product as tensor, is traced as follows, given $R: X \times U \rightarrow Y \times U$, $\trace{U}{X,Y}(x,y) = \exists u. R(x,u,y,u)$.

    The category of finite-dimensional vector spaces is traced as follows.  Given $f: V \tensor U \rightarrow W tensor U$, where $U, V, W$ have bases $\{u_i\}, \{v_j \}, \{ w_k \}$, define 
    \[ \trace{U}{X,Y}(f)(v_j) = \sum_{j,k} a^{ki}_{ji} w_k  \ \ \ \mbox{where} \ \ \ \  f(v_j \tensor u_i) = \sum_{k,m} a^{km}_{ji} w_k \tensor u_m \]

Given $f: X \tensor Y \rightarrow X' \tensor Y$, and $g: Y \tensor Z \rightarrow Y \tensor Z'$, we will in particular consider  $\trace{Y}{X,Z,X',Z'}(f \tensor g):x \tensor Z \rightarrow X' \tensor Z' $\citep{SamsonTraces}.  While the trace operator is symmetric wrt $f,g$, We sometimes write $f; g$ to exploit its compatibility with standard sequential composition, while always being careful to respect the types.

\end{itemize}

\section{Deterministic Games}
In this section, we build a model of IMALL, with all permitted (weak) distributivities. Our development follows the development of~\citet{AJ94}, enhanced with the Question/Answer framework of~\citet{AJMPCF,HYLAND2000285}.  We try to present the model from first principles, but rely extensively on this prior work for the proofs.   

Hopefully, a less onerous way for the non-specialist in game-semantics to read this section is to skip the proofs, and focus on the model and the examples and properties as summarized in the statements of lemmas.  

For the specialist on games, we draw attention to the characterization of composition via the trace operator of~\citet{SamsonTraces} and~\citet{Joyal_Street_Verity_1996}; we reconcile the ``sum interpretation'' with the ''product interpretation'' of the Geometry of Interaction~\citep{haghverdi_geometry_2010}.

\subsection{Games}

A game represents the rules of a discourse between two players, the Player and Opponent. Games are presented as finite sets of finite maximal positions \footnote{Since all positions are finite, there are no winning conditions.}. 
\begin{defn}[Games]
$\agame = ( \amoves,\aPO,\aposnmax, \acr)$ where 
\begin{itemize}
    \item $\aPO: \amoves \rightarrow \{ P, O \} \times \{ Q , A \}$ is a labelling function.
     \item $\acr \subseteq \amoves \times \amoves$ is a binary coherence relation on answers such that if $\amove \acr \bmove$ then $\aPO(\amove) = \aPO(\bmove)$.
    \item $\aposnmax \subseteq \amoves^{\star}$ is the finite set of finite maximal positions of the game such that for all $\apos \in \aposnmax$:
    \begin{description}
    \item[Alternation: ] Moves are of alternating polarity of $O$ and $P$, and the first move is by $O$
    \item[Well-Balanced: ] In any prefix of $\apos$, the number of questions is never less than the number of answers. $\apos$  has equal numbers of questions and answers.
    \item[Completeness:] The set of answers in $\apos$ is the union of a  maximal $P$-clique and a maximal $O$-clique.  
    \item[Extension: ] If $\apos \in \pref{\aposnmax}$, and $\atv$ (resp. $\btv$) is a maximal $O$-clique (resp. maximal $P$-clique) such that $\answers(\apos) \subseteq \atv \cup \btv$, then
    \[ (\exists \apos' \in \aposnmax) \ [ \apos\preceq \apos', \answers(\apos') = \atv \cup \btv ]
    \]
    \end{description}
\end{itemize}
\end{defn}
We drop the subscript identifying the game whenever possible.   We use $\overline{\aPO}$ to indicate the labeling function that flips the $O/P$ polarities.  We use $\answers(\apos)$ to refer to the set of answers in $\apos$.  We use $\oans(S)$  for the set of $O$-answers in $S$).  Similarly, $\pans(S)$ is the set of $P$-answers. 

The games we consider are negative games, where the opening move is by the Opponent.  The moves are labelled as questions and answers.  Viewing questions as $($, and answers as $)$, the completed positions in $\aposnmax$ are well-balanced strings of parenthesis, in which the left parenthesis never get ahead of the right parenthesis.  Thus, there is a natural notion of a matching question for each answer, with every question answered in a complete play, which therefore has even length.    

Answers identify the accumulation of information in a position. Both players can contribute to the accumulation of information via answers. A maximal play represents total information, that is approximated by the partial information in its prefixes.  

The coherence relation constrains the answers that can appear together in a position. Every maximal play contains a complete clique of Opponent answers and a complete clique of Player answers.  Since the first move is an Opponent-question that has to be matched by a Player-answer, any game with a non-empty set of plays has at least one $P$-clique.

\emph{Extension} is a saturation condition.  The partial information in a prefix can always to be completed, \emph{collaboratively} by the two protagonists, to total information that is consistent; thus, any position can always be extended to a maximal position with specified cliques that are consistent with  the answers already in the position. 

The maximal cliques in a completed position are the basis for the Hilbert spaces associated with a game.
\begin{defn}
Let $\aod = \Hi{\{\atv \mid \atv\mbox{ is a maximal $O$-clique} \}}$, and $\apd= \Hi{\{\atv \mid \atv\mbox{ is a maximal $P$-clique}\}}$.  Then:
    \[ \ad = (\aod, \apd)\]
\end{defn}
A theme in our development will be that whereas the Hilbert spaces of the games will identify the additives and the multiplicatives, the game structure provides a more refined perspective that distinguishes them.  

\begin{example}\label{simplegame}
The game $1$ has no moves. 

The game $\bool$ has $O$ question move $\perp$, and $P$ answer moves $\true, \false$,  with $\true \coh\ true$ and $\false \coh\ \false$.  The two maximal positions are $\perp \true, \perp \false$.  
\end{example}

\subsection{Constructions on objects}
\paragraph*{Tensor. }
The tensor construction is standard.  
\begin{defn}[Tensor]
Let $\agame = ( \amoves,\aPO,\aposnmax, \acr)$, and  
$\bgame =  ( \bmoves,\bPO,\bposnmax, \bcr)$.
\[ \agame \tensor \bgame =  ( \amoves \uplus \bmoves,\aPO \uplus\bPO,\posnmax, \coh)\]
where
\begin{itemize}
\item $\coh\ = \ \acr \cup \bcr \cup\ \{ (\amove, \bmove) \mid \amove, \bmove \ \mbox{ answers in different components}  \} $
\item For all $\apos \in \posnmax$, $\apos \rest \agame$ (resp. $\apos \rest \bgame$) is in $\aposnmax$ (resp. $\bposnmax$)
\end{itemize}
\end{defn}
We sometimes refer to $\agame, \bgame$ as the \emph{components}.  The conditions ensure that only $O$ switches between the components $\agame$ and $\bgame$; a constraint that limits the information flow between the components for the Player relative to the Opponent.   It also ensures that the answer matching a question in $\agame$ is from $\agame$; similarly for $\bgame$.   Tensor introduces no new incoherence; answers from different components are always coherent, and the maximal cliques of $\agame \tensor \bgame$ are pairs of maximal cliques of $\agame$ and $\bgame$.  Thus:
\[ \dom{\agame \tensor \bgame} = (\aod \tensor \bod, \apd \tensor \bpd) \]

To verify that the tensor game satisfies \emph{extension}, consider a non-maximal position $\apos$.  Since $O$ can switch components, the key case is when $P$ is to move in $\apos$.  If $P$ is to move in a component (say $\agame$), the projection of $\apos$ to it $(\apos \rest \agame)$ is of odd length, and hence not maximal; so there is an extension of the position with suitable properties from the extension hypothesis on the component game. 

\begin{example}
$\bool_1 \tensor \bool_2$ has $O$ question moves $\perp_1, \perp_2$, and $P$ answer moves $\true_1,\true_2 \false_1,\false_2$ has eight maximal positions; two of them are
\[
\begin{array}{l}
\perp_1 \true_1 \perp_2 \true_2 \\
\perp_2  \true_2 \perp_2 \false_2
\end{array}
\]
The other maximal positions are derived by the symmetry between $\true, \false$ and between $\agame, \bgame$
\end{example}

\paragraph*{Linear Implication. }
Linear implication switches the $P/O$ polarity in $\agame$.   In contrast to the standard treatment, we require both components to be visited in a complete play, following the fair games of Section 6 of~\citet{MURAWSKI2003269}.
\begin{defn}[Linear Implication]
Let $\agame = ( \amoves,\aPO,\aposnmax, \acr)$, $\bgame = ( \bmoves,\bPO,\bposnmax, \bcr)$.
\[ \agame \linimpl \bgame =  ( \amoves \uplus \bmoves,\overline{\aPO} \uplus\bPO,\posnmax, \coh)  \]
where
\begin{itemize}
\item $\coh\ = \ \acr \cup \bcr \cup\ \{ (\amove, \bmove) \mid \amove, \bmove \ \mbox{ answers in different components}  \}$
\item For all $\apos \in \posnmax$, $\apos \rest \agame$ (resp. $\apos \rest \bgame$) is in $\aposnmax$ (resp. $\bposnmax$)
\end{itemize}
\end{defn}
The last move of any maximal position is in $\bgame$; furthermore, this can only happen after reaching a maximal position in $\agame$.   In particular, for any game $\agame$, $\agame \linimpl \one = \one$.  

Only $P$ switches between the components $\agame$ and $\bgame$.  The answer matching a question in $\agame$ is from $\agame$; similarly for $\bgame$.   

Linear implication introduces no new incoherence; answers from different components are always coherent; so, the maximal cliques of $\agame \linimpl \bgame$ are pairs of maximal cliques of $\agame$ and $\bgame$.   The Hilbert space of the implication game incorporates the impact of the inverted polarity in $\agame$:
\[ \dom{\agame \linimpl \bgame} = (\apd \tensor \bod, \aod \tensor \bpd) \]

To verify that the implication game satisfies \emph{extension}, consider a non-maximal position $\apos$.  Since $P$ can switch components, the key case is when $O$ is to move in $\apos$.  There are two cases.
\begin{itemize}
    \item $\apos = \epsilon$.  If $\bgame$ is empty, so is $\agame \linimpl \bgame$ and $\apos$ is maximal.  Otherwise, $O$ can move in $\bgame$. 
    
    \item $\apos = \apos' \amove$. $O$ has to move in the same component  as $\amove$. 

    If $\amove$ is from $\agame$.  Since first move in $\apos \rest \agame$ is by Player, $|\apos \rest \agame|$ is odd; so, not maximal in $\agame$. 

     If $\amove$ is from $\bgame$.  The last move in the maximal positions of the implication game is in $\bgame$.  Since $\apos$ is not maximal, $\apos \rest \bgame$ is not maximal.  
\end{itemize}
Thus, there is an extension of the position with suitable properties from the extension hypothesis on the component game.

\begin{example}
$\bool_1 \linimpl \bool_2$ has $O$ question move $\perp_2$, $P$ question move $\perp_2$, $P$ answer moves $\true_2,\false_2$ and $O$ answer move $\true_1, \false_1$   There are only four maximal positions;
\[
\begin{array}{l}
\perp_1 \perp_2 \true_1  \true_2 \\
\perp_1 \perp_2 \true_1\false_2 \\
\perp_1 \perp_2 \false_1  \true_2 \\
\perp_1 \perp_2 \false_1  \false_2
\end{array}
\] 
\end{example}

\paragraph*{Product. }
The Opponent chooses between the two games in a product.
\begin{defn}[Product]
Let $\agame = ( \amoves,\aPO,\aposnmax, \acr)$, $\bgame = ( \bmoves,\bPO,\bposnmax, \bcr)$.
\[ \agame \llwith \bgame =  ( \amoves \uplus \bmoves,\aPO \uplus\bPO,\aposnmax \uplus \bposnmax, \acr \uplus \bcr)  \]
\end{defn}
The extension property for the product game follows immediately from the assumption that the components $\agame,\bgame$ satisfy extension.   

The product game also illustrates how the extension of a position to a complete position requires collaboration, in this case from the Opponent.  Concretely, in the starting position $\epsilon$, the Opponent has to start in $\bgame$ to advance to a complete position in $\bgame$. 

Answers from different components are incoherent.  So, since the maximal plays are inherited from the component games, the maximal cliques of $\agame \llwith \bgame$ are either from $\agame$ or from $\bgame$.
\[ \dom{\agame \llwith \bgame} = (\apd \times \bod, \aod \times \bpd) \]

\paragraph*{Sum}
In contrast to product, the Player chooses between the two component games, facilitated by a protocol at the beginning of the game.  

\begin{defn}[Sum]
Let $\agame = ( \amoves,\aPO,\aposnmax, \acr)$, and  
$\bgame = ( \bmoves,\bPO,\bposnmax, \bcr)$.
\[ \agame \lplus \bgame =  ( \amoves \uplus \bmoves \uplus \{\perp, l, r \} ,\lambda,\posnmax, \coh )  \]
where
\begin{itemize}

\item \[    \lambda(\amove) = \left\{ 
                                \begin{array}{l}
                                \lambda_{\agame}(\amove), \ \amove \in \amoves \\
                                \lambda_{\bgame}(\amove), \ \amove \in \bmoves \\
                                (P,A),  \  \amove \in \{ l, r \} \\
                                (O,P), \   \amove = \perp
                                \end{array}
                            \right.
    \]
    
\item $\posnmax = \perp\ l \ \aposnmax\ \  \uplus\ \  \perp\ r \ \bposnmax$
\item $\coh\  =  \ \acr \cup \bcr \cup\ \{(l,\amove), (l,l) \mid \amove \in \answers(\agame)  \} \cup \{(r,\amove), (r,r) \mid \amove \in \answers(\bgame) \}  $
\end{itemize}
\end{defn}

The extension property for the product game follows immediately from the assumption that the components $\agame,\bgame$ satisfy extension.    

The sum game also illustrates how the extension of a position to a complete position requires collaboration, in this case from the Player.  After the initial Opponent move, we require the Player to start in $\bgame$ in order to advance to a complete position in $\bgame$, 

Answers from different components are incoherent, as are $l,r$.  Since the maximal plays are inherited from the component game, the maximal cliques of $\agame \lplus \bgame$ are either from $\agame$ or from $\bgame$.  In the former (resp.latter) case, these cliques also contain $l$ (rep. $r$).  Thus: 
\[ \dom{\agame \lplus \bgame} \approx (\aod \lplus \bod, \aod \lplus \bpd) \]
coinciding with product above. 

\begin{example}
$\bool = \one\ \lplus\ \one$
\end{example} 

\subsection{Strategy}
Strategies are deterministic transducers that produce responses based on prior history.  While ~\citet{AJMPCF} permits no history, and~\citet{HYLAND2000285} restricts the visible history, we permit our strategies to view the entirety of the past history.  

In anticipation of upcoming sections, the complete plays of a strategy are endowed with scalars; the reader only interested in game semantics may treat the scalar as always $1$.  So, we will think of a strategy $\astr$ as an element of $\Hi{\aposnmax}$\footnote{We will use our notation that permits us to use set-theoretic operations on such strategies.}.
\begin{defn}[Strategy]
Given $\agame = (\amoves,\aPO,\aposnmax, \acr)$, a strategy $\astr \in \Hi{\aposnmax}$ is such that:
\begin{description} 
\item[Determinacy: ] For all  $\apos \in \in \pref{\astr}$ of even length, if $\apos \amove \in \pref{\aposnmax}$, then there exists a unique $\amove$ such that $\apos \amove \bmove  \in \pref{\astr} $
\item[Monotonicity: ] For all $(\ac,\bpos) \in \astr$, if $\answers(\apos) \subseteq \answers(\bpos)$, then $\answers(\apos\amove) \subseteq \answers(\bpos)$.
\item[$\icoh$-preservation: ] Incoherence is preserved.   
\begin{alignat*}{2}
&\mbox{Forall } \apos \amove, \apos'\amove' \in \pref{\astr} \mbox{ with strictly incoherent $O$-answers } \amove \icoh \amove' \\
&\mbox{Forall complete extensions } \apos \amove \preceq \bpos,  \apos' \amove' \preceq \bpos', \{\bpos, \bpos' \} \in \pref{\astr} \cap \aposnmax \\ 
&\mbox{Exists extending prefixes } \apos \amove \apos_1 \bmove \preceq \bpos,  \apos' \amove'  \apos'_1 \bmove' \preceq \bpos' \\
&\mbox{such that $P$-answers } \bmove \icoh \bmove' \mbox{  are strictly incoherent}.
\end{alignat*}
\end{description}
\end{defn}
Determinacy implies \emph{receptivity} and \emph{response}, i.e. strategies have to respond to all valid Opponent moves.    

Monotonicity and preservation of incoherence preservation are global properties.  

The monotonicity axiom bites when strategy makes an answer move; in this case, the answer is consistent with the answers in every extending completed position in $\astr$.   From an information perspective, this condition ensures that the strategy is a monotonic function from Opponent answers to Player answers.

$\icoh$-preservation is to ensure that strategies are 1-1 (equivalently, that they preserve orthogonality in the future sections).  Given {\em any} two prefixes in the strategy that ending with $\icoh$-Opponent moves, a strategy has to respond with incoherent answers sometime in the future\footnote{A common constraint in game semantics with history-free strategies is to require that the functions determining strategies are history-free, a view that coexists nicely with the preservation of partial injections by composition in the Geometry of interaction.  The preservation of incoherence is a semantic generalization that allows us to include additives. }.

A non-example helps to illustrate $\icoh$-preservation.  Consider the attempted constant $\true: \bool \linimpl \bool = \{ \perp \perp \false \true, \perp \perp \true \true \} $.  The prefixes
$ \perp \perp \true$ and $ \perp \perp \false$ end with $\icoh$-Opponent answers, but their sole extensions in $\true$ end with coherent Player answers.  

Strategies always define a (partial) function from $\aod$ to $\apd$.  
\begin{defn}[(Partial) Function of a Strategy]\label{stratFN} Given $\astr: \agame$, define $\fn{\astr}: \aod \rightarrow \complex \times \apd{} \cup \{ 0 \}$ by
\[ \fn{\astr}(\atv) =  (\ac, \btv), \ \mbox{if } \exists (\ac,\apos) \in \astr, \oans(\apos) = \atv, \pans(\apos) =\btv 
\] 
\end{defn}
$\fn{\astr}$ is well-defined.
\begin{lemma}
$\fn{\astr}: \aod \rightarrow \complex \times \apd$ is a 1-1 (partial) function.
\end{lemma}
\begin{proof}
    To prove uniqueness, let $\apos, \bpos \in \astr$ such that $\oans(\apos) = \oans(\bpos)$.  We will show that $\pans(\apos) = \pans(\bpos)$.  We will show by induction on the prefixes $\apos' \preceq \apos$ that $\pans(\apos') \subseteq \pans(\bpos)$. 
    The base case is immediate.  We proceed by cases based on the parity of $\apos'$.
    \begin{itemize}
    \item If $\apos'$ is of odd length.  By response and determinacy, there is a unique $\amove$ such that $\pref{\astr} \ni \apos' \amove$ and $\apos' \amove \preceq \apos$.  By monotonicity,  $\pans(\apos' \amove) \subseteq \pans(\bpos)$.
    \item If $\apos'$ is of even length.  Then $\apos' \amove \preceq \apos$ is such that $\oans(\apos' \amove) \subseteq \oans(\bpos)$ and $\pans(\apos' \amove) = \pans(\apos') \subseteq \pans(\bpos)$. 
    \end{itemize}
    By completeness condition on maximal positions, $\apos$ has complete cliques of answers; so, we deduce that $\pans(\apos) = \pans(\bpos)$.

     To prove $\fn{\astr}$ is $1-1$, we show that if $\atv \not= \btv$, then $\fn{\astr}(\atv) \not= \fn{\astr}(\btv)$.
Let $\atv \not= \btv$.  Then, there are $\{ \apos \amove, \apos' \amove' \} \in \pref{\astr}, \amove \icoh \amove'$ be such that $\apos \amove \preceq \bpos, \oans(\bpos) = \atv, \oans(\bpos') = \atv$.  Since $\astr$ preserves incoherence, we deduce that $\pans(\bpos) \not=\pans(\bpos')$, yielding the result.
\end{proof}

We consider strategies modulo extensional equivalence.  In this definition, two partial functions are equal if their graphs coincide.  
\begin{defn}[Equivalence of Strategies]
Let $\astr, \bstr: \agame$.  Then $\astr \approx\ \bstr$ if 
$\fn{\astr} = \fn{\bstr}$. 
\end{defn}  

In anticipation of future sections, we visualize a strategy as a matrix $\mat{\astr}$, that we draw as:
\[
\left[
\begin{array}{c|c|c|c}
    & \btv_1 & \cdots & \btv_m \\ \hline
\atv_1 & x_{11} & \cdots & x_{1m} \\ \hline
\vdots & \vdots & \ddots & \vdots \\ \hline
\atv_n & x_{n1} & \cdots & x_{nm}
\end{array}
\right]
\]
Rows are indexed by the basis elements of $\aod$ ($\atv_1 \ldots \atv_n$ are maximal $O$-cliques); columns are indexed by the basis elements of $\apd$ (maximal $P$-cliques $\btv_1 \ldots \btv_m$).  This matrix represents a linear transformation from $\aod$ to $\apd$, and has a special form for the current strategies; no row or column has more than one non zero-entry.   If a maximal $O$-clique is not in the domain, its row has all $0$'s.  

We record a simple observation about strategies. 
\begin{lemma}[Multiplication by scalars]
If $\astr: \agame$, then for any $\ac$,$\ac\ \astr: \agame$.  $\mat{\ac\astr} = \ac\ \mat{\astr}$.  
\end{lemma}

\subsection{The category $\bgg$}
The category $\bgg$ has:
\begin{description}
    \item[Objects: ] Games $\agame$
    \item[Morphisms: ]  $\astr: \agame \rightarrow \bgame$ if $\astr: \agame \linimpl \bgame $. Modulo $\approx$
\end{description}

\subsubsection{Identity}
The identity is given by a copycat strategy~\cite{AJ94}.  We use primed notation on the games to distinguish different copies of the same game. 
\[ \id = \{ (1, \apos) \in \agame \linimpl \agame' \mid \apos \rest \agame = \apos \rest \agame' \} \]

The figure below shows the prefix of a typical play and the function of the strategy.

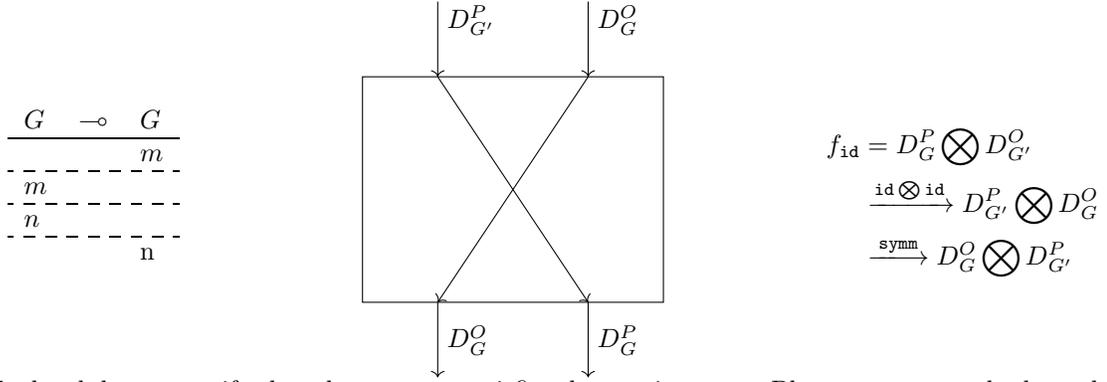
\begin{minipage}{.20\textwidth}
\centering
\begin{tabular}{lll}
$\agame$  &  $\linimpl$ & $\agame $  \\ \hline
&&                       $\amove$   \\ \hdashline
$\amove$ &&  \\ \hdashline
$\bmove$ && \\  \hdashline
&&                       \bmove   \\
\end{tabular}
\end{minipage}
\begin{minipage}{.5\textwidth}
\centering
\begin{tikzpicture}
    % Draw the rectangular box
    \draw (0,0) rectangle (4,3);
    
    % Draw arrows coming in vertically, labeled A and B
    \draw[->] (1,4) -- (1,3) node[near start, right] {$\apdp$};
    \draw[->] (3,4) -- (3,3) node[near start, right] {$\aod$};
    
    % Draw arrows going out, labeled A and B
    \draw[->] (1,0) -- (1,-1) node[midway, right] {$\aod$};
    \draw[->] (3,0) -- (3,-1) node[midway, right] {$\apd$};
    
    % Draw inside the box connections with arrows from A to B, and B to A,
    % ensuring that the lines cross over as requested
    % Connection from A's entry to B's exit
    \draw[->] (1,3) -- (3,0) ;
    % Connection from B's entry to A's exit
    \draw[->] (3,3) -- (1,0);
\end{tikzpicture}
\end{minipage}%
\hfill
\begin{minipage}{.25\textwidth} 
\centering
\begin{alignat*}{3}
\fn{\id} & =  \apd \tensor \aodp  \\
& \xrightarrow{\id \tensor \id} \apdp \tensor \aod  \\
& \xlongrightarrow{{\tt symm}} \aod{} \tensor \apdp
\end{alignat*}
\end{minipage}

We sketch how to verify that the strategy satisfies the requirements.  
Player response only depends on the prior Opponent move; so, $\id$ is history-free.  For all $\apos \in  \pref{\id} $ of even length, 
\[ \apos \rest \agame' = \apos \rest \agame  \]
For all $\apos \in  \pref{\id} $ of odd length,
\[  \apos \rest \agame' \preceq \apos \rest \agame  \ \ \mbox{ or } \ \ \ 
    \apos \rest \agame \preceq \apos \rest \agame'
\]

To show that $\id$ satisfies monotonicity, consider $\apos \amove \in \pref{\id}$  of odd length $\apos \amove \preceq \bpos \in \pref{\id} \cap \aposnmax$.  We consider the case when $\amove \in \agame$ (the case when $\amove \in \agame'$ is symmetric and omitted).

So, $\apos \amove \amove' \in \pref{\id}$, where we use $\amove'$ for the copy of $\amove \in \agame'$.  Then:
\[ \pans(\apos \amove \amove') \rest \agame = \pans(\apos)  \rest \agame \subseteq \pans(\bpos \rest \agame) \]
and
\[ \pans(\apos \amove \amove') \rest \agame' \subseteq \pans(\apos)  \rest \agame' \cup \{ \amove' \} \subseteq \oans(\apos \amove \rest \agame) \subseteq \oans(\bpos \rest \agame) = \pans(\bpos \rest \agame')
\]

To show that $\id$ preserves $\icoh$,  consider $\apos \amove, \bpos \bmove \in \pref{\id}$  of odd length such that $\amove \icoh \bmove$.  Thus, by definition of coherence relation for linear implication, $\amove, \bmove$ are in the same component; furthermore, their copies  $\amove', \bmove'$  in the other component are such $\amove' \icoh \bmove'$.  Result follows since $\apos \amove \amove', \bpos \bmove \bmove' \in \pref{\id}$.

\subsubsection{Composition}
We first describe composition as a solution to a constraint on plays. 
\begin{defn}[Composition]
Let $\astr: \agame \linimpl \bgame$ (resp, $\bstr: \bgame \linimpl \cgame$).   Then, 
\[ \astr ; \bstr = (\astr \pcomp \bstr)\rest A,C \]
where
\[ \astr \pcomp \bstr  = \{ (\ac \bc, \apos) \in  (\amoves \cup \bmoves \cup \cmoves)^{\star} \mid (\ac, \apos) \rest \agame,\bgame \in \astr, (\bc ,\apos )\rest \bgame,\cgame \in \bstr \} \}
\]
\end{defn}

To establish that composition preserves the required properties of strategies, we use the proofs of Propositions 1 of~\citet{AJ94}, and Proposition 2.1 of~\citet{AJMPCF}.  In order to establish a more precise relationship with those papers,
 we observe that:
\[ \pi_2 (\astr \pcomp \bstr) =  \maximal{(\pref{\astr} \pcomp \pref{\bstr})} \]
where
\[ \pref{\astr} \pcomp \pref{\bstr} =  \{ \apos \in  (\amoves \cup \bmoves \cup \cmoves)^{\star} \mid \apos \rest \agame,\bgame \in \astr, \apos \rest \bgame,\cgame \in \bstr \}\]
In particular, We use a key observation from~\citet{AJ94}.    For any $\apos \in \pref{\astr; \bstr}$, there is a unique associated $\apos' \in \pref{\astr \pcomp  \bstr}$ such that $\apos' \rest \agame, \cgame = \apos$.  This shows that there is an iterative way to build up the solution to the constraint outlined in the definition of composition.

\begin{lemma}\label{comp}
Let $\astr \approx \astr': \agame \linimpl \bgame, \bstr \approx \bstr': \bgame \linimpl \cgame$.  Then:
\begin{itemize}
    \item $\astr; \bstr: \agame \linimpl \cgame$ 
    \item $\astr; \bstr \approx \astr'; \bstr': \agame \linimpl \cgame$
\end{itemize}
\end{lemma}
\begin{proof}

The proof that $\astr; \bstr$ satisfies Determinacy follows the proof of Proposition 1 of~\citet{AJ94}, and is omitted.

\paragraph*{Monotonicity. }  Let $\apos \in \pref{\astr; \bstr}$.  Let $\bpos \in \maximal{\pref{\astr; \bstr}}$ such that $\answers(\apos) \subseteq \answers(\bpos)$.  Let $\apos' \in \pref{\astr \pcomp \bstr}$ and $\bpos' \in \pref{\astr \pcomp \bstr}$ be the corresponding witnesses, such that $\apos' \rest \agame,\cgame = \apos$ and $\bpos' \rest \agame,\cgame = \bpos$.  We prove that $\answers(\apos') \subseteq \answers(\bpos')$ by induction on $\apos'$.
Our starting point is with $\epsilon$ where both sets of answers are empty.   
\begin{itemize}
\item When $O$ moves in $\agame,\cgame$, result follows by assumption $\answers(\apos) \subseteq \answers(\bpos)$. 
\item When the next move is by Player, the monotonicity assumption on $\astr, \bstr$, ensure that their responses maintains the invariant.
\end{itemize}

\paragraph*{$\icoh$— preservation.}  We establish a slightly more general result by the showing the result for  positions from $\pref{\astr \pcomp \bstr}$, i.e. Let $ \amove \icoh \amove', \aPO(\amove) = \aPO(\amove') =  OA, \apos \amove \preceq \bpos, \apos'\amove'\preceq \bpos' \in \pref{\astr \pcomp \bstr}$, where $\bpos,\bpos' \in \maximal{\pref{\astr \pcomp \bstr}}$.  

In the coherence relation for linear implication, incoherent moves have to be in the same component.  In particular, $\amove,\amove'$  are locked together in one of $\agame$ or $\bgame$ or $\cgame$.  

Using the incoherence of $\astr,\bstr$, we discover extensions  $\apos\amove \preceq \bpos_1\bmove \preceq \bpos $  and $\apos'\amove'\bpos'_1\bmove' \preceq \bpos'$, such that $\bmove \icoh \bmove'$.  There are two possibilities.
\begin{itemize}
    \item If $\bmove,\bmove'$ are in $\agame$, we are done.  Similarly, if both are in $\cgame$.   
    \item  Otherwise, both moves are in $\bgame$.  In this case, repeat the argument with $\bpos_1\bmove$ and $\bpos'_1\bmove'$.  
\end{itemize}  
The result follows, since the play eventually exits $\bgame$.

\paragraph*{$\approx$. }
Let $\apos \in \astr \pcomp \bstr$.  Let $\astr \approx \astr', \bstr \approx \bstr'$.  We build a witness $\bpos \in \astr' \pcomp \bstr'$ inductively, always maintaining the invariant:
\[ \answers(\bpos) \subseteq \answers(\apos) \]
Our starting point is with $\epsilon$.  The rules we follow are:
\begin{itemize}
\item When $O$ is to move in $\agame,\cgame$, choose a move that results in a position that is consistent with $\answers(\apos)$.  This is always possible by the extensibility axiom on games.  
\item When the next move is by Player, the monotonicity assumption on $\astr, \bstr$, ensure that their responses maintains the invariant. 
\end{itemize}
Since the completed position $\bpos$ also satisfies the invariant, we deduce that $\answers(\apos) = \answers(\bpos)$ and $\astr; \bstr \approx \astr';\bstr'$.

\end{proof}
The above proof incorporates ideas from both the ``sum'' and ''product'' flavors of GOI~\citep{haghverdi_geometry_2010}. Whereas, preservation of incoherence is in the spirit of the ''sum''-version; the history-sensitivity of strategies makes the computation of the fixed point to yield the result to have the ''product'' flavor.

The composition of the function of strategies can be calculated using partial trace operators~\citet{SamsonTraces,Joyal_Street_Verity_1996}\footnote{We encourage the specialists in game semantics to compare to~\citet{McC10}.}as per the following figure, writing $\fn{\astr};\ fn{\bstr}$ for the result.  The game constraints ensure that the witness for the composition of functions is unique!

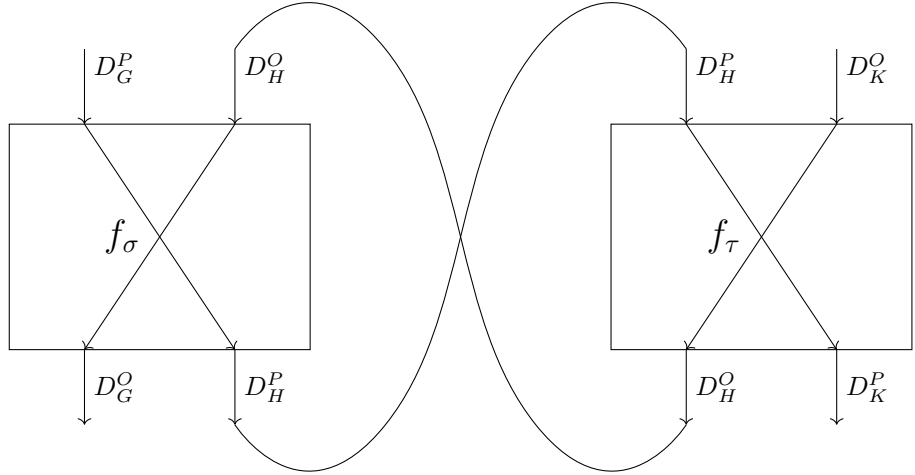
\begin{minipage}{\textwidth}
\centering
\begin{tikzpicture}
    % Diagram on the left
    \begin{scope}
    % Draw the rectangular box
    \draw (0,0) rectangle (4,3);
    \node at (1.5,1.5) {{\Large $\fn{\astr}$}};
    
    % Draw arrows coming in vertically, labeled A and B
    \draw[->] (1,4) -- (1,3) node[near start, right] {$\apd$};
    \draw[->] (3,4) -- (3,3) node[near start, right] {$\bod$};

    % Draw arrows going out, labeled A and B
    \draw[->] (1,0) -- (1,-1) node[midway, right] {$\aod$};
    \draw[->] (3,0) -- (3,-1) node[midway, right] {$\bpd$};
    
    % Draw inside the box connections with arrows from A to B, and B to A,
    % ensuring that the lines cross over as requested
    % Connection from A's entry to B's exit
    \draw[->] (1,3) -- (3,0) ;
    % Connection from B's entry to A's exit
    \draw[->] (3,3) -- (1,0);
     \end{scope}
    
    % Diagram on the right, offset by the width of the left diagram plus a gap
    \begin{scope}[xshift=8cm] % Adjust xshift for the gap between diagrams
         \draw (0,0) rectangle (4,3);
          \node at (1.5,1.5) {{\Large $\fn{\bstr}$}};
    
    % Draw arrows coming in vertically, labeled A and B
    \draw[->] (1,4) -- (1,3) node[near start, right] {$\bpd$};
    \draw[->] (3,4) -- (3,3) node[near start, right] {$\cod$};
    
    % Draw arrows going out, labeled A and B
    \draw[->] (1,0) -- (1,-1) node[midway, right] {$\bod$};
    \draw[->] (3,0) -- (3,-1) node[midway, right] {$\cpd$};
    
    % Draw inside the box connections with arrows from A to B, and B to A,
    % ensuring that the lines cross over as requested
    % Connection from A's entry to B's exit
    \draw[->] (1,3) -- (3,0) ;
    % Connection from B's entry to A's exit
    \draw[->] (3,3) -- (1,0);
    \end{scope}
    
    % Precisely connecting arrows with double bends
    % From the right exit (B's exit) of the left diagram to the left entry (A's entry) of the right diagram
    \draw  plot [smooth, tension=1] coordinates {(9,-1) (7,-1) (5,4)(3,4)};
    \draw  plot [smooth, tension=1] coordinates {(3,-1) (5,-1) (7,4)  (9,4)};
\end{tikzpicture}
\end{minipage}

\begin{lemma}\label{trace}
$\fn{\astr; \bstr} = \fn{\astr}; \fn{\bstr} $
\end{lemma}
\begin{proof}
We are going to prove that
\[  \fn{\astr; \bstr} =  \trace{\bpd,\bod}{\apd\tensor\cod,\aod\tensor\cpd} (\fn{\astr} \tensor \fn{\bstr})\]
The above lemma shows that there are  $\ats_1 \in \aod, \ats_2 \in \cod$ such that:
\[ 
\begin{array}{ll}
\fn{\astr;\bstr}(\ats_1 \tensor \ats_2) =& \ac\bc (\ats'_1 \tensor \ats'_2), \mbox{ where } (\exists \bts' \in \bod,\bts \in \bpd) \\
& \ \ \ \ \  \fn{\astr}(\ats_1 \tensor \bts') =  \ac (\ats'_1 \tensor \bts), \\
& \ \ \ \ \ \fn{\bstr}(\ats_2 \tensor \bts')  = \bc(\ats'_2 \tensor \bts)
\end{array}
\]

Below, we show that the function is well-defined, by showing that there is a unique $\bts,\bts'$ in the above by showing that any values that satisfy the requirements are witnessed by a trace $\apos \in \pref{\astr \pcomp \bstr}$ such that $\answers(\apos) \subseteq \ats_1 \cup \ats'_1 \cup \ats_2 \cup \ats'_2 \cup \bts \cup \bts'$ that we build inductively.

Our starting point is with $\epsilon$.  The rules we follow are:
\begin{itemize}
\item When $O$ is to move in $\agame$, choose a move that results in a position that is consistent with  $\ats_1$.  
When $O$ is to move in $\cgame$, choose a move that results in a position that is consistent with  $\ats_2$. This is always possible by the extensibility axiom on games.  
\item When the next move is by Player, the monotonicity assumption on $\astr, \bstr$, ensure that their responses maintains the invariant.   
\end{itemize}
When $\apos$ is completed, it continues to satisfy the invariant, so $\answers(\apos) = \ats_1 \cup \ats'_1 \cup \ats_2 \cup \ats'_2 \cup \bts \cup \bts'$. 

\end{proof}

As a corollary, since $\fn{\astr; \bstr} = \fn{\astr}; \fn{\bstr}$ is $1-1$, we deduce that
\begin{lemma}\label{trinv}
 $\inv{\fn{\astr; \bstr}} = \inv{\fn{\bstr}} ; \inv{\fn{\astr}}$  is a (partial) $1-1$ function.
 \end{lemma}

\subsection{Strategies for Multiplicatives}\label{smccmor}

The witnesses for the symmetric monoidal structure are:
\begin{eqnarray*}
\twist &=& \{ (1, \apos) \in \agame_1 \tensor \agame_2 \linimpl \agame'_1 \tensor \agame'_2 \mid \apos \rest \agame_1 = \apos \rest \agame'_2, \apos \rest \agame_2 = \apos \rest \agame'_1 \} \\
\assoc &=& \{ (1, \apos) \in \agame_1 \tensor (\agame_2 \tensor \agame_3) \linimpl (\agame'_1 \tensor \agame'_2)  \tensor \agame'_3 \mid \apos \rest \agame_i = \apos \rest \agame'_i, i \in 1 \ldots 3  \} \\
\assoc^{-1} &=& \{ (1, \apos) \in (\agame_1 \tensor \agame_2) \tensor \agame_3 \linimpl \agame'_1 \tensor (\agame'_2  \tensor \agame'_3) \mid \apos \rest \agame_i = \apos \rest \agame'_i, i \in 1 \ldots 3  \} 
\end{eqnarray*}

All the above strategies are history-free copycat strategies like $\id$.  Their functions are similarly constructed from identity functions.  

Let $\astr_i: \agame_i \linimpl \agame'_i$, $i=1,2$.  Then:
\[ 
\astr_1 \tensor \astr_2 = \{ (\ac_1\ac_2, \apos) \in \agame_1 \tensor \agame_2 \linimpl (\agame'_1 \tensor \agame'_2)  \mid (\ac_i, \apos \rest \agame_i \linimpl \agame'_i \in \astr_i, i=1,2 \}
\]
$\astr_1 \tensor \astr_2$ inherits monotonicity and $\icoh$-preservation from $\astr_i$.  
Thus:
\[ \fn{\astr_1 \tensor \astr_2} = \fn{\astr_1} \tensor \fn{\astr_2} \ \ \ \mbox{and} \ \   \mat{\astr \tensor \bstr}  = \mat{\astr} \tensor \mat{\bstr}  \]

The witnesses for the closed structure are as follows.  Let $\sigma: \agame_1 \tensor \agame_2 \linimpl \agame_3$. 
\begin{eqnarray*}
\app &=& \{ (1, \apos) \in (\agame_1 \linimpl \agame_2) \tensor \agame'_1 \linimpl \agame'_2 \mid \apos \rest \agame_i = \apos \rest \agame'_i, i \in 1 \ldots 2 \} \\
\Lambda(\astr) &=& \{ (\ac, \apos) \in \agame_1 \linimpl (\agame_2 \linimpl \agame_3)  \mid (\exists (\ac, \bpos) \in \astr) [ \apos \rest \agame_i = \bpos \rest \agame_i, i=1\ldots 3]  \}
\end{eqnarray*}
$\Lambda(\astr)$ is merely a rebracketing of $\astr$, reflecting the introductory comments on the first-order nature of game semantics.  It inherits  monotonicity and $\icoh$-preservation from $\astr$ and:
\[ \fn{\Lambda{\astr}} = \fn{\astr}, \ \ \ \mbox{and} \ \ \mat{\Lambda(\astr)}  = \mat{\astr} \]

$\app$ is a copy cat; the constraints on positions ensures that the last move in the components when playing $\app$ happen in the order indicated as follows, with the very final move on the right-hand side.

\vspace{5pt}
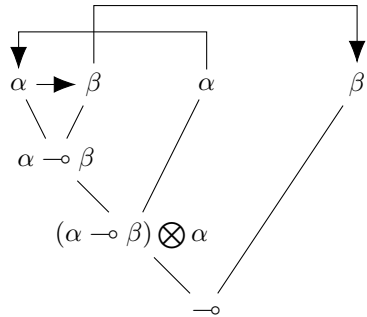
\begin{minipage}{\textwidth}
\centering
\begin{tikzpicture}[node distance=1cm and 1cm]
    \node (a1) at (1,-1.0) {$\linimpl$};  
    \node (a) at (0,0) {$(\alpha\linimpl\beta) \tensor \alpha$};
    \node (b) at (-1,1.0) {$\alpha\linimpl\beta$};
    % \node (c) at (1,1.5) {$\alpha\tensor\alpha$};
    \node (d) at (-1.5,2.0) {$\alpha$};
    \node (e) at (-0.5,2.0) {$\beta$};
    \node (f) at (1,2.0) {$\alpha$};
    \node (g) at (3,2.0) {$\beta$};

    % Lines
    \draw (a1) -- (a);
    \draw (a) -- (b);
    \draw (a) -- (f);
    \draw (b) -- (d);
    \draw (b) -- (e);
    \draw (a1) -- (g); 

% Line segments with three parts
     \draw[-{Latex[length=3mm]}] (d) -- (e);
     \draw[-{Latex[length=3mm]}] (f) -- ++(0, 0.7) -| (d);
     \draw[-{Latex[length=3mm]}] (e) -- ++(0, 1.05) -| (g);
\end{tikzpicture}
\end{minipage}

$\fn{\app}$ follows the recipe for copycats alluded to above.

\subsection{Strategies for Additives}

\paragraph*{Product}
\begin{eqnarray*}
\pi_1 &=& \{ \apos \in (\agame \llwith \bgame)  \linimpl \agame'  \mid \apos \rest \agame_1 = \apos \rest \agame'_1 \} \\
\pi_2 &=& \{ \apos \in (\agame \llwith \bgame)  \linimpl \agame' \mid \apos \rest \agame_2 = \apos \rest \agame'_2 \} 
\end{eqnarray*}
Let $\astr: \agame \linimpl \cgame, \bstr: \bgame \linimpl \cgame$.  Then:
\begin{eqnarray*}
\pair{\astr}{\bstr} &=& \astr \uplus \bstr 
\end{eqnarray*}
\begin{eqnarray*}
\lltwist &=& \{ (1, \apos) \in \agame_1 \llwith \agame_2 \linimpl \agame'_1 \llwith \agame'_2 \mid \apos \rest \agame_1 = \apos \rest \agame'_2, \apos \rest \agame_2 = \apos \rest \agame'_1 \} \\
\llassoc &=& \{ (1, \apos) \in \agame_1 \llwith (\agame_2 \llwith \agame_3) \linimpl (\agame'_1 \llwith \agame'_2)  \llwith \agame'_3 \mid \apos \rest \agame_i = \apos \rest \agame'_i, i \in 1 \ldots 3  \} \\
\inv{\llassoc} &=& \{ (1, \apos) \in (\agame_1 \llwith \agame_2) \llwith \agame_3 \linimpl \agame'_1 \llwith (\agame'_2  \llwith \agame'_3) \mid \apos \rest \agame_i = \apos \rest \agame'_i, i \in 1 \ldots 3  \} 
\end{eqnarray*}
All except pairing are copycat strategies like $\id$.  So, their functions are similarly constructed from identity functions; e.g., see below for $\pi_1$ and pairing.

\begin{minipage}{.42\textwidth}
\renewcommand{\arraystretch}{1.3}
\centering
\[ \mat{\pi_1} \]
\[
\left[
\begin{array}{c|c|c}
        & \aod \tensor \apdp &  \aod \tensor \bpd \\ \hline
\apd \tensor \aodp & \mat{\id_{\agame}} & 0 \\ \hline
\apd \tensor \bpd  &   0                & 0
\end{array}
\right]
\]
\end{minipage}
\begin{minipage}{.42\textwidth}
\renewcommand{\arraystretch}{1.3}
\centering
\[ \mat{\pair{\astr}{\bstr}} \]
\[
\left[
\begin{array}{c|c|c}
        & \aod \tensor \cpd &  \bod \tensor \cpd \\ \hline
\apd \tensor \cod & \mat{\astr} & 0 \\ \hline
\apd \tensor \cod  &   0   & \mat{\bstr}          
\end{array}
\right]
\]
\end{minipage}
\vspace{10pt}

Pairing requires history, as we see from the beginning of a possible play in $\pair{\astr}{\bstr}$ below (assuming that $\bmove, \bmove',\amove,\amove' \in \pref{\astr}$). The last response of the strategy requires memory of the initial Opponent move to know which of $\astr,\bstr$ is being used.  

\begin{minipage}{\textwidth}
\centering
\begin{tabular}{lll}
$\agame  $ & $\linimpl$ & $\bgame \llwith \cgame $ \\ \hline
&&                 $\bmove$  \\ \hdashline
$\bmove'$ & &                    \\ \hdashline
$\amove$ &&                         \\ \hdashline
&&                       $\amove'$  \\
\end{tabular}
\end{minipage}

The distributivity strategies are copycats, albeit requiring history. 
\begin{defn}[$\llwith$-distributivity]
Let $\bgame = \agame_1 \linimpl (\agame_2 \llwith \agame_3)$ and let $\bgame' = (\agame'_1 \linimpl \agame'_2) \llwith (\agame'_1 \linimpl \agame_3) $
\begin{eqnarray*}
\distiw &=& \{ \apos \in \bgame \linimpl \bgame' \mid \apos \rest \agame_1, \agame_2 = \apos \rest \agame'_1, \agame'_2 \} \cup
\{ \apos \in \bgame \linimpl \bgame' \mid \apos \rest \agame_1, \agame_3 = \apos \rest \agame'_1, \agame'_3 \} \\
\inv{\distiw} &=& \{ \apos \in \bgame' \linimpl \bgame \mid \apos \rest \agame_1, \agame_2 = \apos \rest \agame'_1, \agame'_2 \} \cup
\{ \apos \in \bgame \linimpl \bgame' \mid \apos \rest \agame_1, \agame_3 = \apos \rest \agame'_1, \agame'_3 \} \\
\end{eqnarray*}
There is one way distributivity with tensor; $\agame \tensor (\bgame \llwith \cgame) \linimpl (\agame \tensor \bgame) \llwith (\agame_2 \tensor \bgame)$ defined as:
\[ \{ \apos \mid \apos \rest (\agame \tensor (\bgame \llwith \cgame) = \apos \rest  (\agame \tensor \bgame) \} \cup 
\{ \{ \apos \mid \apos \rest (\agame \tensor (\bgame \llwith \cgame) = \apos \rest  (\agame \tensor \cgame) 
\}
\]
\end{defn}

\paragraph*{Sum}
The strategies involving $\lplus$ have to account for the initial protocol.

\begin{defn}\label{summor}  
    \begin{eqnarray*}
    \injl &=& \{ (1, \perp\ l \ \apos) \in \agame \linimpl \agame' \lplus \bgame \mid \apos \rest \agame = \apos \rest \agame' \} \\
    \injr &=& \{ (1, \perp\ r \ \apos) \in \bgame \linimpl \agame \lplus \bgame' \mid \apos \rest \bgame = \apos \rest \bgame' \} 
    \end{eqnarray*}
    Let $\astr: \agame \linimpl \cgame, \bstr: \bgame \linimpl \cgame$.  Then:
    \begin{eqnarray*}
    \lsumstr{\astr}{\bstr} &=& \{ \amove \perp l \ \apos \mid \amove \apos \in \astr\} \cup \{ \amove \perp\ r \ \apos \mid \amove \apos \in \bstr\}
    \end{eqnarray*}
    \begin{eqnarray*}
    \ptwist &=& \{ (1, \apos) \in \agame_1 \lplus \agame_2 \linimpl \agame'_1 \lplus \agame'_2 \mid \apos \rest \agame_1 = \apos \rest \agame'_2, \apos \rest \agame_2 = \apos \rest \agame'_1 \} \\
\passoc &=& \{ (1, \apos) \in \agame_1 \lplus (\agame_2 \lplus \agame_3) \linimpl (\agame'_1 \lplus \agame'_2)  \lplus \agame'_3 \mid \apos \rest \agame_i = \apos \rest \agame'_i, i \in 1 \ldots 3  \} \\
\inv{\passoc} &=& \{ (1, \apos) \in (\agame_1 \lplus \agame_2) \lplus \agame_3 \linimpl \agame'_1 \lplus (\agame'_2  \lplus \agame'_3) \mid \apos \rest \agame_i = \apos \rest \agame'_i, i \in 1 \ldots 3  \} 
       \end{eqnarray*}
\end{defn}
All the above strategies except $\lsumstr{\cdot}{\cdot}$ are copycats, including on the important protocol bits (namely $l,r$) that choose the summand.  So, their functions are similarly constructed from identity functions. They all require history, in part to navigate the bits of protocol.  The beginning of a typical play in $\injl: \agame \linimpl \agame \lplus \bgame$ and in $[\astr,\bstr]:  \agame \lplus \bgame \linimpl  \cgame $  are described below. In $\lsumstr{\astr}{\bstr}$, after the prefix of the first three moves that chooses $l$, the play stays in $\astr$.  In the figure below,  we assume that $\amove\amove' \in \pref{\astr}$.  

\begin{minipage}{.5\textwidth}
\centering
\begin{tabular}{lll}
$\agame$  &  $\linimpl$ & $\agame\lplus   \bgame$ \\ \hline
&&                       $\perp$  \\ \hdashline
&&                       l        \\ \hdashline
&&                       \amove   \\ \hdashline
$\amove$ && \\ \hdashline
$\bmove$ && \\ \hdashline
&&                       \bmove   \\
\end{tabular}
\end{minipage}
\begin{minipage}{.5\textwidth}
\centering
\begin{tabular}{lll}
$\agame   \lplus      \bgame$ & $\linimpl$ & $\cgame $ \\ \hline
&&                 $\perp$  \\ \hdashline
$\perp$ & &                    \\ \hdashline
$l$ &&                         \\ \hdashline
$\amove$ && \\ \hdashline
&& $\amove'$  
\end{tabular}
\end{minipage}

\noindent 
The function for the injections (resp. sum) coincide with those for projections (resp. product).  Compare the figure below with earlier  figure for product.  Games differentiate the two additives via the Protagonist who chooses the component, thus the polarity prevents the identification of the additives.  

\begin{minipage}{.5\textwidth}
\renewcommand{\arraystretch}{1.5}
\centering
\[
\mat{\injl} \]
\[
\left[
\begin{array}{c|c|c}
        & \aod \tensor \apdp &  \aod \tensor \bpd \\ \hline
\apd \tensor \aodp & \mat{\id_{\agame}} & 0 \\ \hline
\apd \tensor \bpd  &   0                & 0
\end{array}
\right]
\]
\end{minipage}
\begin{minipage}{.5\textwidth}
\renewcommand{\arraystretch}{1.5}
\centering
\[
\mat{\lsumstr{\astr}{\bstr}}
\]
\[
\left[
\begin{array}{c|c|c}
        & \aod \tensor \cpd &  \bod \tensor \cpd \\ \hline
\apd \tensor \cod & \mat{\astr} & 0 \\ \hline
\apd \tensor \cod  &   0   & \mat{\bstr}          
\end{array}
\right]
\]
\end{minipage}

\vspace{5pt}

\begin{defn}[$\tensor-\lplus$ distributivity] \label{distmor} \hfill \\
Let $\bgame= \agame_1 \tensor (\agame_2 \lplus \agame_3)  $, $\bgame' = (\agame'_1 \tensor \agame'_2) \lplus (\agame'_1 \tensor \agame'_3) $
\begin{eqnarray*}
    \distts &=& \{ (1, \apos) \in \bgame \linimpl  \bgame'  \mid \apos \rest \agame_i = \apos \rest \agame'_i, i=1\ldots 3, \\
    && l \in \apos \rest \bgame \leftrightarrow l \in \apos \rest \bgame', r \in \apos \rest \bgame \leftrightarrow r \in \apos \rest \bgame' \} \\ 
    \distts^{-1} &=& \{ (1, \apos) \in \bgame' \linimpl  \bgame  \mid \apos \rest \agame_i = \apos \rest \agame'_i, i=1\ldots 3, \\
    && l \in \apos \rest \bgame \leftrightarrow l \in \apos \rest \bgame' , r \in \apos \rest \bgame \leftrightarrow r \in \apos \rest \bgame' \} \\
\end{eqnarray*}
\end{defn}
$\distts,\inv{\distts}$ are copycats on each component and on the protocol bits.   So, their functions are similarly constructed from identity functions. Both need history, as by the following plays. 
\begin{table}[ht]
\centering
\begin{tabular}{lllllll}
$\agame_1$&$\tensor$&$\agame_2\ \lplus \agame_3$ & $\linimpl$ & $\agame'_1 \tensor \agame'_2$& $\lplus$& $\agame'_1 \tensor \agame'_3$ \\ \hline
&&&&&$\perp$ & \\ \hdashline
&& \ \ \ \ \   $\perp$ & & &                   \\ \hdashline
&& \ \ \ \ \  \   $l$ &&   &&                      \\ \hdashline
&&&&& \ $l$  \\ \hdashline
&&&& $\amove$  \\ \hdashline
$\amove$ &&                    \\ \hdashline
$\amove'$ &&\\ \hdashline
&&&& $\amove'$  \\ \hdashline
&&&& \ \ \ \ \ \ \ \ \  $\bmove$ \\ \hdashline
&& \    $\bmove$ &&   &&                      \\ 
\end{tabular}
\end{table}
\begin{table}[ht]
\centering
\begin{tabular}{lllllll}
$\agame'_1 \tensor \agame'_2$& $\lplus$& $\agame'_1 \tensor \agame'_3$ & $\linimpl$ & $\agame_1$&$\tensor$&$\agame_2\ \lplus \agame_3$  \\ \hline
&&&&$\amove$ \\ \hdashline
&\ $\perp$&\\ \hdashline
&\ \ $l$ \\ \hdashline
$\amove$ \\ \hdashline
$\amove'$ \\ \hdashline
&&&&$\amove'$ \\ \hdashline
&&&&&& \ \ \ \ \ \ $\perp$ \\ \hdashline
&&&&&& \ \ \ \ \ \ \ $l$ \\ \hdashline
&&&&&& $\bmove$ \\ \hdashline
\ \ \ \ \ \ \ \ \  $\bmove$ \\
\end{tabular}
\end{table}

\subsection{Properties of $\bgg$}

\begin{lemma}
$(\bgg,\one,\tensor,\linimpl)$ is symmetric monoidal closed.  Sums are given by $\lplus$ and products are given by $\llwith$.  Tensor distributes over sum and linear implication distributes over product.  
\end{lemma}
The associativity of composition follows from Proposition 2 of ~\citet{AJ94}.  
The witnesses required for the structure are given by definitions~\ref{smccmor},~\ref{summor}, and~\ref{distmor}. Lemma~\ref{comp} permits an easy verification of required equations by considering the functions associated with the strategies.  

\section{ The category $\vg$}\label{sec:vg}

Following the motivation in the introduction for considering sums, we introduce the category $\vg$ with morphisms as sums of strategies from $\bgg$.  The strategies from $\bgg$ thus provide a basis for the strategies in $\vg$.  

\begin{defn}[Category $\vg$]
The objects in $\vg$ are games.

$\avec: \vg(\agame, \bgame)$ is an element of $ \Hi{\posnmax_{\agame \linimpl \bgame}}$ such that:
            \[ (\exists \astr_i \in \bgg(\agame,\bgame) \ \avec  = \sum_i \ac_i \astr_i \]
\end{defn}
There is no reason for the witnesses to be unique in the above definition.    In the technical development of this and future section, all our definitions will continue to be phrased as before on '' strategies as sets of positions"; the existence of a witness for the decomposition into sums will be treated merely as an extra constraint that our definitions have to enforce.  

We will follow the convention of using $\avec, \bvec$ for sums of morphisms in $\bgg$, reserving $\astr,\bstr$ for individual morphisms from $\bgg$.  In general, the decomposition into sums is not unique.   We will only ever rely on the existence of \emph{some} decomposition. 

\begin{defn}
Given $\avec \in \Hi{\posnmax_{\agame}}$, define $\fn{\avec}$ as the unique linear  extension of 
\[ (\forall \atv \in \aod)  \fn{\avec}(\atv) = \sum \ac_i \btv_i \mid (\exists (\ac_i, \apos) \in \avec, \answers(\apos) = \atv \cup \btv  \]
\end{defn}
The function associated with $\avec$, as defined above, is independent of the choice of the decomposition of $\avec$ into morphisms of $\bgg$.   However, if $\avec  = \sum_i \ac_i \astr_i$, then
  \[ \fn{\avec} = \sum_i \ac_i \fn{\astr_i} \]

As before, we consider morphisms up to $\approx$-equivalence.
\begin{defn}
   Given $\avec,\bvec:  \vg(\agame,\bgame)$, define 
   $\avec \approx \bvec$ if $\fn{\avec} = \fn{\bvec}$
\end{defn}
From the results forthcoming in section~\ref{schwinger}, we will deduce that $\vg(\ndim,\ndim)$ and $\vg(\bool^n,\bool^n)$ realize all linear transformations.  

Before we define composition, we generalize restriction to the case of sums.   
\begin{defn}[Restriction]
Let $S \subseteq (\amoves)^{\star}$, $\bmoves \subseteq \amoves $.  Then
\[ S \rest \bmoves = \sum \ac\apos \mid (\exists (\ac,\bpos) \in S, \bpos = \apos \rest \bmoves \]
\end{defn}
Restriction, when used on \emph{sets of positions}, can cause positions to cancel because the scalars add to $0$.  This is a key feature of quantum phenomenon (see example~\ref{sqnot} below) that does not appear in the deterministic setting of the prior sections, where the above sum had only one summand.  

Composition is defined as before; we reiterate that this definition does not rely on a choice of witnesses.  
\begin{defn}[Composition]
Let $\avec: \agame \linimpl \bgame$, $\bvec: \bgame \linimpl \cgame$.   Then, 
\[ \avec ; \bvec = (\avec \pcomp \bvec)\rest A,C \]
where
\[ \avec \pcomp \bvec = \sum \ac \ac' \apos \mid \apos \in  (\amoves \cup \bmoves \cup \cmoves)^{\star},  \ac (\apos \rest \agame,\bgame) \in \avec, \ac'( \apos \rest \bgame,\cgame) \in \bvec 
\]
\end{defn}

Composition can also be defined in terms of particular witnesses. 
\begin{lemma}\label{sumcomp} If 
$ \avec  = \sum_i \ac_i \astr_i$ and $  \bvec  = \sum_j \bc_j  \astr_i$, then
\[ \avec ; \bvec = \sum_{i,j} \ac_i \bc_j\  \astr_i; \bstr_j \]
where the composition on the LHS is in $\vg$, and the composition on the right is in $\bgg$.
\end{lemma}
\begin{proof}
It suffices to show that 
\[ \avec \pcomp \bvec =\sum_{i,j} \ac_i \bc_j \ \astr_i \pcomp \bstr_j \]
where
\[ \avec =  \sum_i \ac_i \astr_i,  \ \ \ \   \bvec = \sum_i \bc_j \bstr_j \]
which follows by noting the bi-linearity and commutativity properties of $\pcomp$.
\[ \avec \pcomp \sum_j \bc_j \bvec_j = \sum_j \bc_j (\avec \pcomp \bvec_j)  \]
and
\[ \avec \pcomp \bvec =  \bvec \pcomp \avec\]
\end{proof}

\begin{lemma}\label{vgtrace}
$\fn{\avec; \bvec} = \fn{\avec}; \fn{\bvec} $
\end{lemma}
\begin{proof}
Under the hypothesis of lemma~\ref{sumcomp},
\begin{alignat*}{3}
\fn{\avec ; \bvec} &= \sum_{i,j} \ac_i \bc_j \fn{\astr_i; \bstr_j}  \\
&= \sum_{i,j} \ac_i \bc_j \fn{\astr_i}; \fn{\bstr_j} \\
&= \sum_{i,j} \fn{\ac_i\astr_i}; \fn{\bc_j \bstr_j} \\
&= (\sum_i\fn{\ac_i\astr_i}) ; (\sum_j \fn{\bc_j \bstr_j})  \ \ \ \  \mbox{    Linearity of trace} \\
&= \fn{\avec}; \fn{\bvec}
\end{alignat*}
\end{proof}

While the definitions of composition in $\vg$ is similar to the deterministic categories, quantum effects arise from the interference across summands.   
\begin{example}\label{sqnot}[$\sqrt{{\tt NOT}}; \sqrt{{\tt NOT}} = {\tt NOT}$]
As anticipated in the introduction, sums permit us to split the symmetry of all monoidal structures.  We demonstrate for the tensor product below.  Let $\ac = \frac{1}{2}( 1 +i)$.  Consider
\[ \astr = \sqrt{{\tt NOT}}: \bool \linimpl \bool  =  \ac\ \id + \conj{\ac} \ {\tt NOT} \]
Thus:
\[ \astr = \{ (\ac, \perp \perp \true\ \true), 
              (\ac, \perp \perp \false\ \false),
               (\conj{\ac}, \perp \perp \false\ \true), 
                (\conj{\ac}, \perp \perp \true\ \false) \}
\]
We are going to compute: 
\[ \bool_1 \xlongrightarrow{\displaystyle{\astr}} \bool_2  \xrightarrow{\displaystyle{\astr}} \bool_3\]
$\astr \pcomp \astr$ is
\[
\begin{array}{lll}
&\{ &  (\ac\ac, \perp_3 \perp_2  \perp_1 \true_1\  \true_2\ \true_3), 
              (\ac\ac, \perp_3 \perp_2 \perp_1 \false_1\ \false_2\ \false_3), \\
    && (\ac\conj{\ac},  \perp_3 \perp_2 \perp_1 \true_1\ \false_2\ \false_3)
               (\ac\conj{\ac},  \perp_3 \perp_2 \perp_1 \false_1\ \true_2\ \true_3), \\
    &&  (\conj{\ac}\conj{\ac}, \perp_3 \perp_2 \perp_1 \true_1\ \false_2\ \true_3), 
               (\conj{\ac}\conj{\ac}, \perp_3 \perp_2 \perp_1 \false_1\ \true_2\ \false_3) \\
    &&(\ac\conj{\ac},  \perp_3 \perp_2 \perp_1 \true_1\ \true_2\ \false_3)
               (\ac\conj{\ac},  \perp_3 \perp_2 \perp_1 \false_1\ \false_2\ \true_3 \\
    &\}& 
\end{array} 
\]
Restricting out the middle $\bool_2$ to compute $\astr;\astr$ yields:
\[ \astr ; \astr = \{ (\ac\ac +\conj{\ac}\conj{\ac} , \perp_3 \perp_1 \true_1\true_3), 
              (\ac\ac +\conj{\ac}\conj{\ac} , \perp_3 \perp_1 \ \false_1 \false_3),
              (2 \ac\conj{\ac} ,  \perp_3  \perp_1 \true_1 \false_3),
              (2 \ac\conj{\ac} ,  \perp_3  \perp_1 \false_1\true_3)
               \}
\]
Since $\ac^2 +(\conj{\ac})^2 = 0$, and $2 \ac\conj{\ac} = 1$ we are left with:
\[ \astr ; \astr = \{
              (1,  \perp_3 \perp_1 \true_1\ \false_3),
              (1,  \perp_3  \perp_1 \false_1\ \true_3)
               \} = {\tt NOT}
\]
\end{example}

By standard considerations, $\vg$ essentially has (non-empty) biproducts,   So, we focus on $\lplus$ from now on.

\begin{lemma}
$\vg$ is symmetric monoidal-closed with sums where tensor distributes over sum.
\end{lemma}
\begin{proof}
$\bgg$ is a sub-category of $\vg$.  The required witnesses for identity, symmetry, associativity, injections, and distributivity (including weak distributivity) of the monoidal structures, are all deterministic copy cats, and inherited from $\bgg$.

$\Lambda(\cdot), \pair{\cdot}{\cdot}, \lsumstr{\cdot}{\cdot}, \tensor$ are defined similarly to $\bgg$. 
\begin{eqnarray*}
\avec_1 \tensor \avec_2 &=& \{ (\ac_1\ac_2, \apos) \in \agame_1 \tensor \agame_2 \linimpl (\agame'_1 \tensor \agame'_2)  \mid (\ac_i, \apos \rest \agame_i \linimpl \agame'_i \in \avec_i, i=1,2 \} \\
\Lambda(\avec) &=& \{ (\ac, \apos) \in \agame_1 \linimpl (\agame_2 \linimpl \agame_3)  \mid (\exists (\ac, \bpos) \in \avec) [ \apos \rest \agame_i = \bpos \rest \agame_i, i=1\ldots 3]  \} \\
\lsumstr{\avec}{\bvec} &=& \{ \amove \perp l \ \apos \mid \amove \apos \in \avec\} \cup \{ \amove \perp\ r \ \apos \mid \amove \apos \in \bvec\} 
\end{eqnarray*}
These are linear in each argument, e.g.
\begin{eqnarray*}
{\mathlarger{\mathlarger{\Lambda}}} (\ac_i \sum_i \avec_i) &=& \sum_i \ac_i \Lambda(\avec_i) \\
\left[\sum_i \ac_i \avec_i,\sum_j \bc_j \bvec_j \right] &=& \sum_{i,j} \ac_i \bc_j \lsumstr{\avec_i}{\bvec_j} \\
\left(\sum_i \ac_i \avec_i) \tensor ( \sum_j \bc_j \bvec_j \right)  &=& \sum_{i,j} \ac_i \bc_j  (\avec_i \tensor \bvec_j) 
\end{eqnarray*}
\end{proof}

\paragraph*{Discussion }

The category $\vg$ is closely related to the one derived from the GOI-construction on finite-dimensional vector spaces with tensor products, as discussed in ~\citet{ABRAMSKY20031}. Specifically, the composition of functions on strategies within $\vg$ aligns with their definition of composition. This makes $\vg$ a refinement of their category, specifically tailored to games. While their original category is compact-closed, $\vg$ does not exhibit $\star$-autonomy due to its restriction to games initiated by the Opponent. Moreover, $\vg$ distinguishes between the two multiplicatives based on considerations of the direction of information flow, which are inherent to the game structure.

This refined approach is exploited in subsequent sections, where we focus $\vg$ on reversible and unitary strategies. Within these discussions, the uniqueness of composition witnesses, guaranteed by lemma~\ref{trace} and lemma~\ref{trinv}, allow us to demonstrate that reversibility and unitarity properties are conserved under composition.

\section{Reversibility}

We consider a subcategory of reversible morphisms in $\bgg$.  Recall that ``reversible computing does not necessarily mean computations are invertible''~\citep{HEUNEN2015217}.  So, we do not demand that $\astr: \bgg(\agame,\bgame)$ is an isomorphism in $\bgg$.  Rather, we proceed as follows.

\begin{defn}
$\astr:\bgg(\agame, \bgame)$ is reversible if 
\[ \adj{\fn{\astr}} = \inv{\fn{\astr}} \]
\end{defn}

Notable non-reversible morphisms are the injections and projections, as we can deduce from cardinality constraints.  Many of the other strategies that we have seen in $\bgg$ are reversible.    These include:
\begin{itemize}
\item Identity
\item Symmetry, Associativity of all monoidal structures
\item Distributivities, including weak distributivities
\item If $\astr,\bstr$ are reversible, so are $\astr \tensor \bstr, \pair{\astr}{\bstr}, \lsumstr{\astr}{\bstr}, \Lambda(\astr)$
\end{itemize}

Reversible morphisms are closed under composition.  We show that
\[ \adj{\fn{\astr;\bstr}} = \inv{\fn{\astr;\bstr}} \]
appealing to  lemma~\ref{trace} below in the second step.
\begin{alignat*}{2}
& \fn{\astr;\bstr}(\ats_1 \tensor \ats_2) = \ac\bc\  \ats'_1 \tensor \ats'_2  \\
\Longleftrightarrow & (\exists \ac',\bc') \ (\exists \bts' \in \bod,\bts \in \bpd) \fn{\astr}(\ats_1 \tensor \bts') =  \ac' (\ats'_1 \tensor \bts), \fn{\bstr}(\ats_2 \tensor \bts') =   \bc'(\ats'_2 \cup \bts), \ac\bc = \ac'\bc'  \\
\Longleftrightarrow & (\exists \ac',\bc') \ (\exists \bts' \in \bod,\bts \in \bpd) \inv{\fn{\astr}}(\ats'_1 \tensor \bts) =  \conj{\ac'} (\ats_1 \tensor \bts'), \inv{\fn{\bstr}}(\ats'_2 \tensor \bts) = \conj{\bc'}(\ats_2 \tensor \bts'), \ac\bc = \ac'\bc' \\
\Longleftrightarrow &  (\exists \ac',\bc') \ (\exists \bts' \in \bod,\bts \in \bpd) \adj{\fn{\astr}}(\ats'_1 \tensor \bts) =  \conj{\ac'} (\ats_1 \tensor \bts'), \adj{\fn{\bstr}}(\ats'_2 \tensor \bts) = \conj{\bc'}(\ats_2 \tensor \bts'), \ac\bc = \ac'\bc'  \\
\Longleftrightarrow & \conj{(\ac\bc)} \  \ats_1 \tensor \ats_2 = \adj{\fn{\bstr}}; \adj{\fn{\astr}}
\end{alignat*}
proving that:
\[ \adj{\fn{\astr;\bstr}} = \adj{\fn{\bstr}}; \adj{\fn{\astr}} \]
Result follows by noting:
\[
\inv{\fn{\astr;\bstr}} = \inv{\fn{\bstr}}; \inv{\fn{\astr}} 
= \adj{\fn{\bstr}}; \adj{\fn{\astr}}
\]

\subsection{The Schwinger basis}\label{schwinger}

We explore the coding of the Schwinger basis~\citet{Schwinger} (See ~\citet{Wheeler} for a survey of several options of unitary bases)  using reversible strategies.

As a preliminary exercise, we show that any 1-1 boolean operator is realizable.  

\begin{lemma}[Universality for Boolean operators] \label{perm}
Any $1-1$ operator on $\bool^n \linimpl \bool^n$ is realized by a strategy. 
\end{lemma}
\begin{proof}
Given a $1-1$ operator $f$ on $\bool^n$, we define $\astr$ such that $\fn{\astr}=f$ as follows.
\begin{itemize}
    \item After initial query on right-hand side, say in $i$'th component, query all the arguments (yielding say $\vec{b}$), produce an answer $\pi_i(f(\vec{b})$ in $i$'th component.
    \item For following queries, say in $j$'th component  on the right-hand side, use the history to respond immediately with $\pi_j(f(\vec{b})$ in $j$'th component.
\end{itemize}
Since the function is $1-1$, any $O$-established incoherence (necessarily in the games on the LHS) is witnessed by $P$-established incoherence (necessarily in the games on the RHS), at one of the components on the RHS.

\end{proof}

\paragraph*{Completeness: $\ndim$.} 
The prototypical $n$-dim vector space is $\ndim =\lplus_{0..n-1} \one$.  The Schwinger basis of unitary operators on $\ndim$ is given by the compositions of the following operators; drawn here below for $n=5$.  On the left is the one-step permutation operator, and in the right, $\omega$ is the fifth root of unity.  

\begin{tabular}{lll} \\
$ \left[
\begin{tabular}{l|l|l|l|l}
0 & 0 & 0 & 0 & 1 \\ \hline
1 & 0 & 0 & 0 & 0 \\ \hline
0 & 1 & 0 & 0 & 0 \\ \hline
0 & 0 & 1 & 0 & 0 \\  \hline
0 & 0 & 0 & 1 & 0 \\ 
\end{tabular}
\right]
$
&\hspace*{1in}&
$ \left[
\begin{tabular}{l|l|l|l|l}
1 & 0 & 0 & 0 & 0 \\ \hline
 0& $\omega$ & 0 & 0 & 0 \\ \hline 
0 & 0 & $\omega^2$ & 0 & 0 \\ \hline 
0 & 0 & 0& $\omega^3$  & 0 \\ \hline
0 & 0 & 0 & 0 & $\omega^4$ \\ 
\end{tabular} 
\right] $
\\
\end{tabular}

Lemma~\ref{perm} shows us how to code any $1-1$ function on basic datatype as a $\bgg$ strategy.  
The matrix on the left is realized by the $\bgg$ strategy:
\[ \ndim \xlongrightarrow{\displaystyle{[\id, \omega\ \id, \omega^2\ \id, \omega^3\ \id, \omega^4\ \id]}} \ndim \]

\paragraph*{Completeness. }
Consider $\bool^n \linimpl \bool^n$.    We describe how to build the required witnesses for the Schwinger basis in $\bgg$.  The matrices in this case are $2^n \times 2^n$.

The permutation matrices are $1-1$ truth tables on $n$-bits.  So, by lemma~\ref{perm}, we have an associated  $\bgg$ strategy.  

We build the required diagonal Schwinger matrix, using morphisms available in $\bgg$ as follows.  Let $\omega$ be $n$'th root of unity.  By induction on $j$, we build a strategy $\astr: \bgg(\bool^j, \bool^j)$ with the diagonal entries of $\mat{\astr}$ being $1, \omega, , \ldots, \omega^{j-1}$: 
\[
\left[
\begin{array}{c|c|c|c}
 1 & 0& \cdots & 0 \\ \hline
 0 & \omega& \cdots & 0 \\ \hline
 \vdots & \vdots& \ddots & \vdots \\ \hline
 0 & 0 & \cdots & \omega^{j-1}
\end{array}
\right]
\]
The base case, $j=1$  is realized by the strategy  $\astr_1$
\[ \{ (1, \perp\perp\true\ \true), (\omega,  \perp \perp  \false\ \false) \}\]
For the inductive step $j+1$, given $\astr_j: \bool^j \linimpl \bool^j$ from inductive hypothesis, build $\astr_{j+1}$.
\begin{eqnarray*}
\astr_{j+1} &:& \bool^{j+1} \\
&=&  (\one \lplus \one) \tensor \bool^j \\
&\xrightarrow{\displaystyle{\distts}} & (\one \tensor \bool^j) \lplus  (\one \tensor \bool^j) \\
&\xrightarrow{\displaystyle{\astr_j \tensor \omega^j\ \astr_j}}&  (\one \tensor \bool^j) \lplus  (\one \tensor \bool^j)  \\
&\xlongrightarrow{\displaystyle{\inv{\distts}}}&  (\one \lplus \one) \tensor \bool^j = \bool^{j+1}
\end{eqnarray*}

\subsection{Universality for reversible circuits}   
Given reversible $\astr: \bgg(A,A)$ , reversible $\cntrl \astr$ is given by:
\begin{eqnarray*}
\cntrl \astr    &:& \bool \tensor T \\
&=&  (\one \lplus \one) \tensor T \\
&\xrightarrow{\displaystyle{\distts}} & (\one \tensor T) \lplus  (\one \tensor T) \\
&\xrightarrow{\displaystyle{\astr\tensor \id}}&  (\one \tensor T) \lplus  (\one \tensor T)  \\
&\xlongrightarrow{\displaystyle{\inv{\distts}}}&  (\one \lplus \one) \tensor T
\end{eqnarray*}
Since the Toffoli gate is $ \cntrl (\cntrl \twist)$, this provides an alternate proof that reversible morphisms are universal for reversible circuits.

\section{The category $\qg$}

The category $\qg$ restricts to those strategies in $\vg$ whose functions are unitary, and  which can be decomposed as sums of reversible strategies. 
\begin{defn}[$\qg$]
 $\qg$ is the subcategory of $\vg$ with morphisms $\avec: \vg(\agame,\bgame))$ in $\qg$ if 
\[ (\exists \emph{ reversible } \astr_i \in \bgg(\agame,\bgame) \ \avec  = \sum_i \ac_i \astr_i \ \ \ \ \ 
\mbox{       and        } \ \  \ \ \
\adj{\fn{\avec}} = \inv{\fn{\avec}} \]
\end{defn}

Notable non-unitary strategies are the injections $\injl, \injr$ , since the functions of these strategies are not total in $\bgg$.  Consequently, $\qg$ does not have sums, and $\lplus$ is available only as monoidal structures.  Thus, it is perhaps more appropriate to term $\one \lplus \one$ as $\qbit$ (rather than $\bool$), in the context of $\qg$.  

Copycat strategies from $\bgg$ are morphisms in $\qg$ in the case where their functions are total and the dimensions of the Hilbert spaces of Opponent and Player answers are equal.   

A sufficient condition to verify the condition on dimensions is that the underlying formula is balanced, i.e., the number of positive and negative occurrences of every atom is equal.  An easy verification shows that this condition encompasses:
\begin{itemize}
\item Identity 
\item Symmetry, Associativity of all monoidal structures ($\tensor, \lplus$)
\end{itemize}
For the following strategies, an ad hoc check verifies that the dimension condition is satisfied.  
\begin{itemize}
\item Distributivity of $\tensor$ over $\lplus$)
\item Weak distributivity of $\linimpl$ over $\lplus$
\end{itemize}
Furthermore, if $\astr,\bstr$ are unitary, so are 
\[ \astr \tensor \bstr, \pair{\astr}{\bstr}, \lsumstr{\astr}{\bstr}, \Lambda(\astr) \]

Thus, $\qg$ has all the required ingredients to be symmetric monoidal closed with ($\tensor, \linimpl$); it also supports another symmetric monoidal structures ($\lplus$) and inherits associated distributivity (tensor over sum) from $\vg$.   To complete this picture, we verify closure under composition below.
\begin{lemma}
Let $\avec:\qg(\agame,\bgame), \bvec: \qg(\bgame,\cgame)$.  Then:  $\avec; \bvec: \qg(\agame,\bgame)$.
\end{lemma}
\begin{proof}
Let $ \avec = \sum_i \astr_i, i =1 \ldots m$, $\bvec = \sum_j \bstr_j, j =1 \ldots m$, where $\astr_i,\bstr_j$ are in $\bgg$.  

In the calculation below, we use linearity of $\adj{\cdot}$, bilinearity of composition, $\adj{(\fn{\astr_i})}= \inv{\fn{\astr_i}}$, $\adj{(\fn{\bstr_j})}= \inv{\bstr_j}$, and 

\begin{alignat*}{3}
\adj{\fn{\avec;\bvec}} &= \adj{(\sum_{i,j} \fn{\astr_i;\bstr_j}) } \\
&= \sum_{i,j} \adj{\fn{\astr_i;\bstr_j}} \\
&= \sum_{i,j} \inv{\fn{\astr_i;\bstr_j}} \ \ \ \ &  (\astr_i;\bstr_j, \mbox{ reversible } \in \bgg(\agame,\cgame)\\
&= \sum_{i,j} \inv{\fn{\bstr_j}};\inv{\fn{\astr_i}}; \ \ &(\astr_i;\bstr_j \mbox{ reversible } \in \bgg(\agame,\cgame) \\ 
&=  (\sum_j \inv{\fn{\bstr_j}}); (\sum_{i} \inv {\fn{\astr_i}}) \\
&= (\sum_j \adj{\fn{\bstr_j}}); (\sum_{i} \adj{\fn{\astr_i}});  &(\astr_i;\bstr_j \mbox{ reversible } \in \bgg(\agame,\cgame)   \\
&= \adj{\fn{\bvec}}; \adj{\fn{\avec}} \\
&= \inv{\fn{\bvec}}; \inv{\fn{\avec}}  \ \ \ \ (\avec,\bvec \in \qg) \\
&= \inv{\fn{\avec; \bvec}}  
\end{alignat*}
\end{proof}

\paragraph*{Universality.  }
We conclude this section by demonstrating universality in a couple of different ways.  

For our first proof, we appeal to section~\ref{schwinger}, where we showed that the Schwinger basis for $n$-ary Boolean is realized by morphisms in $\bgg$.  Since every unitary $n$-ary Qbit operation can be expressed as a (weighted) sum of Schwinger basis elements, we deduce universality of $\qg$ unitary $n$-ary Qbit operations.  

We explore an alternate proof of universality by exploring the relationship with ~\citet{10.1145/3632861}.  Consider $\R$, the subcategory of $\qg$ as follows:
\begin{itemize}
\item Objects coincide with $\qg$
\item $\astr: \R(\agame,\bgame)$ is an isomorphism in $\qg$
\end{itemize}
$\R$ includes:  
\begin{itemize}
\item Identity
\item Symmetry, Associativity of both monoidal structures ($\tensor, \lplus$)
\item Distributivity ($\tensor$ over $\lplus$)
\end{itemize}
Furthermore, if $\astr,\bstr$ are isomorphisms, so are 
\[ \astr \tensor \bstr, \lsumstr{\astr}{\bstr} \]
Thus, $\R$ is a $(\tensor,\lplus)$ rig-groupoid.  Furthermore, $\R$ includes 
$\sqrt{X}: \qg(\bool ,\bool)$ as defined by $\frac{1}{2} [ (1+i) \id + (1-i) \twist] $; and all single qbit rotations are available in $\qg$ since the morphisms in $\bgg$ are closed under scaling by $\ac$ such that $|\ac|=1$.  Thus, $\R$ satisfies all requirements of~\citet{10.1145/3632861}. 

$R$ does not include $\app$, $\Lambda(\astr)$, and weak distributivity, since they are not invertible.  So, $\R$ can be viewed as a “first order'' fragment of $\qg$.  In this first-order fragment, reversibility and isomorphisms are identified.   

\paragraph{Acknowledgements.  }   Radha Jagadeesan thanks Prakash Panangaden and Peter Selinger for several useful discussions.  

\bibliography{sample}
\end{document}